\documentclass[runningheads]{llncs}
\usepackage[T1]{fontenc}
\usepackage{graphicx}
\usepackage{graphicx}     
\usepackage{array}        
\usepackage{booktabs}     
\usepackage{longtable}    
\usepackage{tabularx}     
\usepackage{etoolbox}     
\usepackage{caption}      
\usepackage{supertabular}  

\usepackage[dvipsnames]{xcolor} 
\usepackage{soul}               
\sethlcolor{yellow}
\setstcolor{red}
\setulcolor{green}

\usepackage{url}
\usepackage{hyperref}
\hypersetup{
    colorlinks=true,
    linkcolor=blue,
    citecolor=blue,
    urlcolor=blue
}
\usepackage{tikz}
\usetikzlibrary{tikzmark}

\usepackage{tcolorbox}
\tcbuselibrary{breakable, skins}

\newtcolorbox[auto counter, number within=section]{myframebox}[1][]{
    colback=white,
    colframe=black!80,
    width=\linewidth,
    boxrule=0.5pt,
    arc=2mm,
    breakable,
    enhanced jigsaw,
    fonttitle=\bfseries,
    coltitle=black,
    title=Box \thetcbcounter,
    attach boxed title to top left={xshift=3mm, yshift*=-\tcboxedtitleheight/2},
    boxed title style={
        colback=white!80,
        colframe=black!80,
        boxrule=0.5pt,
        arc=2mm
    },
    #1
}

\newcolumntype{C}[1]{>{\centering\arraybackslash}p{#1}}

\usepackage[lambda,
    advantage,
    operators,
    sets,
    adversary,
    landau,
    probability,
    notions,
    logic,
    ff,
    mm,
    primitives,
    events,
    complexity,
    oracles,
    asymptotics,
    keys]{cryptocode}

\begin{document}
\title{Universally Composable Termination Analysis of Tendermint}

%
%
%
%
%

\author{
Zhixin Dong \inst{1}\orcidID{0009-0003-2979-1125} \and 
Xian Xu \inst{1}\orcidID{0000-0001-9713-9751} \and 
Yuhang Zeng \inst{1}\orcidID{0009-0007-5114-1363} \and 
Mingchao Wan \inst{2} \and 
Chunmiao Li \inst{2}
}


\institute{
East China University of Science and Technology, Shanghai, China \\ 
\email{dongzhixin@mail.ecust.edu.cn} \\ 
\email{xuxian@ecust.edu.cn} \\ 
\email{y80240060@mail.ecust.edu.cn} 
\and
Beijing Academy of Blockchain and Edge Computing, Beijing, China (100080) \\
\email{chainmaker@baec.org.cn} \\
\email{chunmiaoli1993@gmail.com}
}

\maketitle              
\begin{abstract}

Modern blockchain systems operating in adversarial environments require robust consensus protocols that guarantee both safety and termination under network delay attacks. Tendermint, a widely adopted consensus protocol in consortium blockchains, achieves high throughput and finality. However, previous analysis of the safety and termination has been done in a standalone fashion, with no consideration of the composition with other protocols interacting with it in a concurrent manner. 
Moreover, the termination properties under adaptive network delays caused by Byzantine adversaries have not been formally analyzed.  This paper presents the first universally composable (UC) security analysis of Tendermint, demonstrating its resilience against strategic message-delay attacks. By constructing a UC ideal model of Tendermint, we formalize its core mechanisms: phase-base consensus procedure, dynamic timeouts, proposal locking, leader rotation, and others, under a network adversary that selectively delays protocol messages.
Our main result proves that the Tendermint protocol UC-realizes the ideal Tendermint model, which ensures bounded termination latency, i.e., guaranteed termination, even when up to $f<n/3$ nodes are Byzantine (where $n$ is the number of nodes participating in the consensus), provided that network delays remain within a protocol-defined threshold under the partially synchronous net assumption. 
Specifically, through formal proofs within the UC framework, we show that Tendermint maintains safety and termination. By the composition theorem of UC, this guarantees that these properties are maintained when Tendermint is composed with various blockchain components. 

\end{abstract}
\section{Introduction}\label{sec:introduction}

Modern blockchain systems operate in decentralized, adversarial environments that pose fundamental challenges to consensus, where network delays and malicious actors jeopardize protocol correctness and progress. Byzantine Fault-Tolerant (BFT) protocols emerged from Lamport et al.'s formulation of the Byzantine Generals Problem \cite{lamport1982byzantine}, which originally assumed a fully synchronous network model with known message delay bounds. To better reflect real-world conditions, Practical Byzantine Fault Tolerance (PBFT) \cite{castro1999practical} adopts the partially synchronous network model introduced by Dwork et al. \cite{dwork1988consensus}, which presents an unknown Global Stabilization Time (GST) after which all honest parties receive messages within a bounded delay. This model relaxes synchrony assumptions while preserving deterministic guarantees, enabling key properties: security ensures that all honest parties agree on a single value in each consensus round, preventing forks and conflicting decisions; termination guarantees that every honest party reaches a final decision within a finite time bound. However, large-scale blockchain deployments, with thousands of geographically dispersed nodes and unpredictable network conditions, expose limitations in BFT designs. These protocols find it difficult to satisfy the global bounded-delay assumption in highly dynamic networks, which may cause termination to fail. In particular, adversaries can launch network-delay attacks, potentially causing forks, transaction censorship, or complete termination breakdown.

Over the past decades, researchers have developed various BFT protocols to enhance performance and address these challenges. PBFT~\cite{castro1999practical} employs a three-phase (Pre-Prepare, Prepare, Commit) broadcast mechanism to achieve consensus with $O(n^2)$ message complexity, where $n$ denotes the total number of replicas (i.e., distributed participants or nodes maintaining system state copies and engaging in consensus). It guarantees correctness provided the number of Byzantine faults $f$ satisfies $f \le n/3$, assuming that after Global Stabilization Time (GST), messages between honest replicas incur at most a known one-way network delay, denoted by $\Delta$. Under these conditions, the leader collects $2f + 1$ Prepare certificates to drive the system to commit. Zyzzyva~\cite{kotla2007zyzzyva} introduces speculative execution, allowing clients to optimistically accept replies if no conflicting certificates appear; its termination proof still relies on eventual synchrony to abort or commit speculative requests. MinBFT~\cite{veronese2009minimal} reduces communication complexity to $O(n)$ by leveraging trusted monotonic counters, yet its termination proof similarly assumes bounded message delays among honest replicas. Tendermint~\cite{kwon2014tendermint,buchman2016tendermint,buchman2018latest} structures voting into three phases (\text{PROPOSE}, \text{PREVOTE}, and \text{PRECOMMIT}) and employs a lock-based mechanism: once a block is locked, any subsequent conflicting proposal is rejected, ensuring eventual convergence after GST~\cite{milosevic2009unifying}. HotStuff~\cite{yin2019hotstuff} further refines this with pipelined commits and threshold signatures, proving that if three consecutive views proceed without view-changes under network delays below $\Delta$, any two committed blocks lie on the same branch. Overall, these protocols establish that under eventual synchrony, an honest leader gathers enough votes to progress within bounded time, ensuring termination once the network stabilizes.

Although Byzantine Fault Tolerant (BFT) protocols have advanced significantly, existing analyses remain protocol-centric and often overlook behavior under concurrent execution and adversarial delay attacks. This limits their applicability in real-world deployments where multiple protocols interact and selective delays occur. To address this gap, we adopt the Universal Composability (UC) framework~\cite{canetti2001universally}, which defines security via an indistinguishability experiment between a real world, where parties run the protocol under adversary $\mathcal{A}$ and environment $\mathcal{Z}$, and an ideal world, where parties interact solely with a trusted functionality $\mathcal{F}$ mediated by a simulator $\mathcal{S}$ that replicates adversarial effects. A protocol is UC-secure if no environment $\mathcal{Z}$ can distinguish these executions. This guarantee extends to arbitrary concurrent composition, enabling modular construction of complex systems from UC-secure subprotocols without sacrificing termination.

UC has been successfully applied to diverse cryptographic primitives, including secure multi-party computation~\cite{canetti2002universally}, key exchange~\cite{canetti2002universallykeyexchange}, signatures~\cite{canetti2004universally}, commitments~\cite{canetti2001universallycommitments}, and zero-knowledge proofs~\cite{camenisch2011framework}, demonstrating its broad utility. However, to our knowledge, few works conduct UC-based analyses for BFT consensus protocols. In particular, no prior work rigorously models Tendermint within the UC framework, hindering evaluation of its security and termination when composed with other blockchain modules. This work addresses two key questions: (1) How to formally model a BFT protocol as a UC ideal functionality explicitly incorporating delay attacks? (2) What insights does a UC analysis provide on timeout strategies and network assumptions to ensure termination under adversarial delays? Using Tendermint as a case study, we establish a UC model and prove its UC security concerning termination.

\subsection{Related Work}\label{sec:relatedwork}

\subsubsection{BFT Protocol Improvements}

The study of Byzantine Fault Tolerance (BFT) begins with Lamport et al.'s foundational work, ``The Byzantine Generals Problem,'' which establishes the resilience bound of $n > 3f$ replicas under synchronous networks~\cite{lamport1982byzantine}. PBFT~\cite{castro1999practical} extends this to partially synchronous settings with unknown Global Stabilization Time (GST), introducing a three-phase protocol (Pre-Prepare, Prepare, Commit) with $O(n^2)$ message complexity. Subsequent protocols aim to reduce latency, communication overhead, and cryptographic costs while preserving security and termination under eventual synchrony. Zyzzyva~\cite{kotla2007zyzzyva} employs speculative execution to enable fast tentative commits with rollback security. MinBFT~\cite{veronese2009minimal} leverages trusted monotonic counters to simplify view changes and reduce protocol state. SBFT~\cite{gueta2019sbft} integrates threshold signatures, vote aggregation, and fast-path optimizations to achieve subquadratic complexity. Recent designs adapt to hybrid fault models and relaxed synchrony: XFT~\cite{liu2016xft} tolerates simultaneous crash and Byzantine faults under weak timing assumptions; LibraBFT~\cite{baudet2019state}, built on HotStuff, streamlines consensus with modular threshold signatures and pipelined view-changes; BEAT~\cite{duan2018beat} enhances throughput and resilience via cryptographic randomness in geo-distributed deployments.

Despite design diversity, these protocols share a common proof strategy: under eventual synchrony, once honest replicas’ messages are delivered within bounded delay, a designated leader collects $2f+1$ matching votes to finalize a decision. Invariants on locked or committed values typically ensure security to prevent conflicting commits, while termination relies on leader rotation and message delivery guarantees to ensure progress.

\subsubsection{Blockchain and Consortium-Chain Applications}

The rise of blockchain systems drives the adaptation of BFT protocols to large-scale, permissioned deployments spanning thousands of nodes. Tendermint~\cite{kwon2014tendermint,buchman2016tendermint,buchman2018latest} employs a gossip-based peer-to-peer layer combined with a Propose-Prevote-Precommit scheme, achieving linear communication complexity and lock-based security under partial synchrony. Hyperledger Fabric~\cite{androulaki2018hyperledger} builds on PBFT and Kafka-style ordering, decoupling consensus engines from smart contract execution to enhance modularity in enterprise chains. HotStuff’s design supports systems like Diem (formerly Libra)~\cite{yin2019hotstuff,baudet2019state}, using threshold-signature aggregation and pipelining to sustain high throughput across hundreds of validators. Complementary projects such as BFT-SMaRt~\cite{bessani2014state} and MirBFT~\cite{stathakopoulou2019mir} optimize performance under realistic workloads: the former offers a Java library with dynamic reconfiguration and batching, while the latter improves throughput via decoupled networking and formally verified pipelining.

\subsubsection{UC Security Analyses for Consensus}

Traditional BFT security and termination proofs analyze protocols in isolation, often overlooking challenges of concurrent execution and composability in complex blockchain systems. The Universal Composability (UC) framework~\cite{canetti2001universally} overcomes these limitations by modeling protocols as ideal functionalities and requiring real-world executions to be indistinguishable from their ideal counterparts for any environment and adversary. Its composability theorem guarantees that security properties hold when multiple protocols are composed. UC-based analyses have been widely applied to public-chain consensus: Shoup et al.~\cite{shoup2024theoretical} provide a complete UC proof for an asynchronous common subset protocol adapted to permissionless settings; Ciampi et al.~\cite{ciampi2024universal} formalize a transaction serialization mechanism ensuring order fairness in proof-of-work systems; and Badertscher et al.~\cite{badertscher2024bitcoin} develop a composable UC abstraction of the Bitcoin ledger for modular security reasoning. Despite these advances, a UC-based analysis of BFT-style protocols in consortium chains remains absent.

To our knowledge, no prior work fully formalizes a consortium-chain BFT protocol, such as Tendermint, as a UC ideal functionality capturing termination under adaptive network-delay attacks. Existing proofs for Tendermint primarily focus on isolated protocol behavior and do not account for concurrent interactions with transaction execution or network layers. This gap motivates our contribution: a UC modeling of Tendermint and a refined termination proof resilient to adaptively delayed, round-specific adversarial strategies.

\subsection{Our Contributions}\label{sec:contributions}

This work presents the first universally composable (UC) security proof for the termination guarantees of a BFT consensus protocol under adaptive network latency attacks. Using Tendermint as a case study, we formalize its execution within the UC framework, where the interfaces and behavioral constraints of all parties, environment, proposer, validators, and adversary, are rigorously defined. Our model captures key protocol mechanisms such as phase-based consensus (Propose, Prevote, Precommit), proposal-locking, rotating leadership, and timeout adaptation, while explicitly simulating delay attacks by the adversary. We construct a simulator that interacts with the ideal functionality $\mathcal{F}_{\mathsf{Tendermint}}$ on behalf of the adversary, and prove that the real-world protocol $\pi_{\mathsf{Tendermint}}$ is indistinguishable from the ideal execution to any environment.

Our main result shows that in a system with $n \ge 3f+1$ nodes and up to $f$ Byzantine faults, Tendermint guarantees termination by $T^\ast = \mathsf{GST} + O(f^2\Delta)$ for every block height $h$, once the network reaches partial synchrony after Global Stabilization Time (GST) with maximum message delay $\Delta$. We define an ideal functionality $\mathcal{F}_{\mathsf{Tendermint}}$ and prove that $\pi_{\mathsf{Tendermint}}$ UC-realizes it. This implies that by time $T^\ast$, every honest party decides a unique and consistent value $v_h$, maintaining termination even under adversarial delays. Overall, our work contributes a rigorous foundation for analyzing Byzantine consensus under timing attacks, and offers practical insights for securing blockchain deployments in latency-sensitive, heterogeneous environments.

\subsection{Organization} \label{sec:organization}
This paper comprises six sections. Following the current introduction, Section~\ref{sec:preliminaries} establishes foundational knowledge on the Universal Composability (UC) framework and the Tendermint protocol's five-stage voting mechanism (NewRound, Propose, Prevote, Precommit, Commit) with adaptive timeouts. Section~\ref{sec:model} formalizes the partially synchronous network and adversarial delay attack model, defining the ideal functionality $\mathcal{F}_{\mathsf{Tendermint}}$ via cryptographic primitives and time-aware controls. Section~\ref{sec:protocol} details the real-world protocol $\pi_{\mathsf{Tendermint}}$ implementation, demonstrating how it UC-realizes $\mathcal{F}_{\mathsf{Tendermint}}$. Section~\ref{sec:securityProof} provides formal proofs of bounded termination time while preserving termination under $f < n/3$ Byzantine nodes and delay attacks. Finally, Section~\ref{sec:conclusion} summarizes contributions and discusses future directions.

\section{Preliminaries}\label{sec:preliminaries}

\subsection{Universally Composable Security}
The Universally Composable (UC) framework, introduced by Canetti~\cite{canetti2001universally} in 2001, provides a rigorous, simulation-based methodology for defining and proving the security of cryptographic protocols. It ensures that a protocol retains its security properties even when composed with arbitrary other protocols running concurrently. In adversarial network environments, this framework is particularly valuable for modeling complex, interactive systems because it guarantees that modular protocol components preserve their security properties when integrated into larger systems.

\subsubsection{Interactive Turing Machines (ITMs)}

The Universal Composability framework models all protocol participants, adversaries, and ideal functionalities as \emph{Interactive Turing Machines} (ITMs), each instantiated as an \emph{Interactive Turing Machine Instance} (ITI) uniquely identified by $(\mathsf{pid},\mathsf{sid})$, where $\mathsf{pid}$ denotes the party’s logical identity and $\mathsf{sid}$ the protocol session. This pairing enables concurrent executions with session isolation and composability. Each ITI operates with multiple tapes: a read-only identity tape storing its code and identifiers; an output tape for sending messages (including recipient and payload); an input tape for received protocol messages; a subroutine-output tape for functionality responses; a backdoor tape through which the adversary injects messages; and an activation tape controlling when the ITI is scheduled to run. This tape-based structure precisely captures complex interactions and adversarial influences within the UC model.

\subsubsection{The Basic UC Framework}

In the basic UC framework, security is defined via a simulation-based paradigm rooted in computational indistinguishability. Each protocol participant, honest or adversarial, is modeled as a probabilistic polynomial-time interactive Turing machine (ITM). Honest parties follow the protocol $\pi$, while the adversary $\mathcal{A}$ may corrupt a subset of them, gaining full control over their internal state and communication. An external environment machine $\mathcal{E}$ provides inputs to honest parties, observes their outputs, and interacts freely with $\mathcal{A}$.

In the ideal world, honest parties relay their inputs to a trusted ideal functionality $\mathcal{F}$, which carries out the task faithfully (e.g., key exchange, fair computation). For every real-world adversary $\mathcal{A}$, there must exist a simulator $\mathcal{S}$ in the ideal world that, interacting with both $\mathcal{F}$ and $\mathcal{E}$, reproduces the view of the real execution. The protocol $\pi$ is deemed UC-secure if no PPT environment $\mathcal{E}$ can distinguish, with more than negligible advantage, between the real-world execution and the ideal one. Although the basic UC model restricts the environment $\mathcal{E}$ to a single protocol session, thus limiting its ability to capture shared setups like a common reference string (CRS) or public key infrastructure (PKI), it still provides strong security guarantees. 

\begin{definition}[UC-Emulation]
A protocol $\pi$ \emph{UC-emulates} another protocol $\phi$ if, for every probabilistic polynomial-time (PPT) adversary $\mathcal{A}$, there exists a PPT simulator $\mathcal{S}$ such that for all constrained environments $\mathcal{E}$, the following computational indistinguishability holds:
\[
\mathsf{EXEC}_{\pi, \mathcal{A}, \mathcal{E}} \approx \mathsf{EXEC}_{\phi, \mathcal{S}, \mathcal{E}},
\]
where $\approx$ denotes computational indistinguishability with respect to the security parameter.
\end{definition}

\subsubsection{The Generalized UC Framework}

The Generalized UC (GUC) framework extends the basic UC model to faithfully capture scenarios in which multiple protocol instances share a common global setup, such as a public key infrastructure (PKI), a common reference string (CRS), or a public ledger. This extension was first formalized by Canetti et al.~\cite{canetti2007universally} to model protocols that rely on shared services. Unlike the constrained environment of the basic model, where the environment may invoke only a single instance of the protocol under analysis, the GUC environment is \emph{unconstrained}. In practice, this means that the environment $\mathcal{E}$ can launch multiple, possibly concurrent, sessions of the protocol $\pi$ being studied, as well as any auxiliary protocols that depend on the same global state. To represent such shared services, GUC introduces \emph{global ideal functionalities} $\mathcal{F}^{\mathrm{g}}$, which exist independently of any individual session and can be accessed by all parties and adversarial components at any time. 

Formally, we say that a protocol $\pi$ \emph{GUC-emulates} another protocol $\phi$ if, for every probabilistic polynomial-time (PPT) adversary $\mathcal{A}$ operating in the real world, there exists a PPT simulator $\mathcal{S}$ in the ideal world such that no unconstrained environment $\mathcal{E}$ can distinguish between the real-world execution of $\pi$ with $\mathcal{A}$ and the ideal-world execution of $\phi$ with $\mathcal{S}$, even when both executions share the same global functionalities. Concretely, this indistinguishability requirement is stated as follows:

\begin{definition}[GUC-Emulation]
A protocol $\pi$ \emph{GUC-emulates} a protocol $\phi$ if, for every PPT adversary $\mathcal{A}$, there exists a PPT simulator $\mathcal{S}$ such that for all (unconstrained) environments $\mathcal{E}$,
\[
\mathsf{GEXEC}_{\pi, \mathcal{A}, \mathcal{E}} \approx \mathsf{GEXEC}_{\phi, \mathcal{S}, \mathcal{E}},
\]
where $\approx$ denotes computational indistinguishability with respect to the security parameter.
\end{definition}

Despite the greater power afforded to the environment in GUC, the composition theorem remains valid, namely, if a higher-level protocol $\rho$ ordinarily invokes an ideal functionality $\phi$ as a subroutine, and if $\pi$ GUC-emulates $\phi$ with respect to the same global functionalities, then substituting every invocation of $\phi$ in $\rho$ with $\pi$ yields a new protocol $\rho^{\phi\to\pi}$ that GUC-emulates $\rho$.  

\begin{figure}[!t]
  \centering
  \includegraphics[width=\linewidth,keepaspectratio]{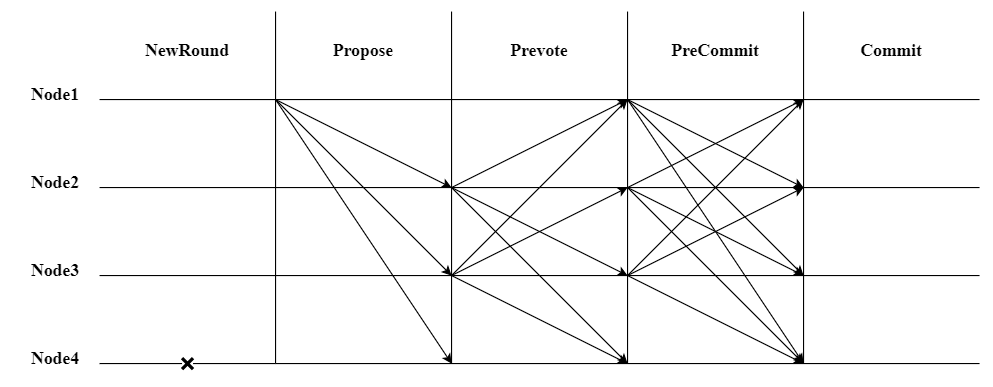}
  \caption{The five phases of Tendermint consensus.}
  \label{Tendermint-phase}
\end{figure}

\subsection{Tendermint Consensus}

The Tendermint consensus algorithm, used by the Chainmaker blockchain, is a Byzantine Fault Tolerant (BFT) protocol for consortium environments. It operates under the partial synchrony model~\cite{dwork1988consensus} with an unknown Global Stabilization Time, GST. The protocol can tolerate up to f Byzantine nodes, under the condition that f is less than one-third of the total nodes. It employs a five-phase, round-based protocol consisting of \emph{NewRound} for leader election via round-robin, Propose for the leader to broadcast a block proposal, \emph{Prevote} for validators to vote on the proposal's validity, \emph{PreCommit} which is triggered upon receiving a two-thirds quorum of votes, and \emph{Commit} for final block finalization. Termination is ensured through timeout-triggered round advancement, and the protocol achieves deterministic termination with a worst-case latency of $O(f^2\Delta)$ after GST, where $\Delta$ represents the bounded network delay. The design incorporates transaction-level validation and randomized transaction selection to provide fair leader rotation and deterministic execution, as illustrated in Fig.~\ref{Tendermint-phase}.

\section{System Model and Ideal Functionality}\label{sec:model}

\begin{figure}[!t]
  \centering
  \includegraphics[width=\linewidth,keepaspectratio]{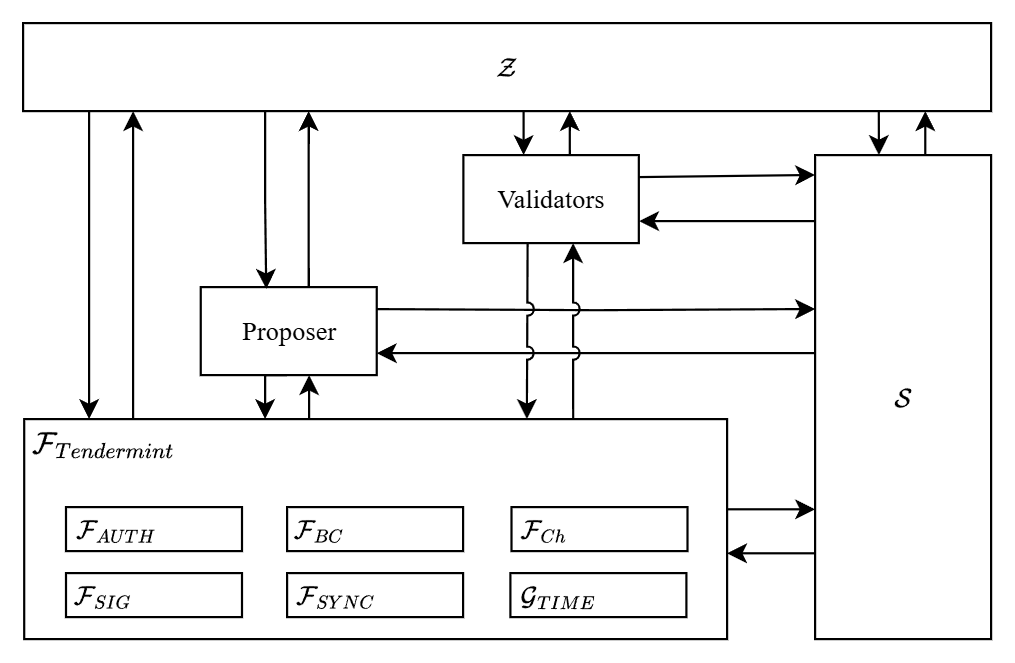}
  \caption{Interactions between Tendermint participants in the ideal-world execution, including their interfaces with the UC simulator.}
  \label{Tendermint-framework}
\end{figure}

\subsection{Tendermint Protocol Participant Roles}

In the Tendermint consensus protocol of Chainmaker, nodes are divided into two core roles: Proposer and Validator. These roles cooperate to ensure that the protocol safely terminates in Byzantine and delay-prone environments. The overall interaction between these roles and their corresponding interfaces with the ideal functionality $\mathcal{F}_{\mathsf{Tendermint}}$ in the UC framework is illustrated in the system architecture (see Fig.~\ref{Tendermint-framework}).

\subsubsection{Proposer and Validator}

At the start of each round, the proposer constructs and digitally signs a block proposal containing transactions, the previous block’s hash, and metadata, broadcasting it reliably with phase-specific timeouts that increase linearly per round. Validators participate in the \emph{Prevote}, \emph{PreCommit}, and \emph{Commit} phases, maintaining local state to guarantee termination under partial synchrony. During Prevote, validators verify block validity and broadcast either \text{PREVOTE}(id(v)), indicating a prevote in favor of the proposal $v$, or \text{PREVOTE}(\text{nil}), indicating a prevote for no proposal; upon receiving $2f+1$ matching prevotes, they lock the block by broadcasting \text{PRECOMMIT}(id(v)) and updating \text{lockedValue} and \text{lockedRound}, otherwise they precommit nil after timeout. Observing a quorum of valid prevotes sets \text{validValue} and \text{validRound} to prioritize future proposals. In the Commit phase, once $2f+1$ matching \text{COMMIT} messages are received, the block is finalized via $\mathcal{\rho}_{\scriptscriptstyle \mathsf{COMMIT}}$ and the height advances, with supermajority commits ensuring prompt finalization.

\begin{center}
\centering
\begin{myframebox}[title=Timer Functionality $\mathcal{G}_{\mathsf{TIME}}$]
The functionality $\mathcal{G}_{\mathsf{TIME}}$ models a logical countdown timer. It is parameterized over a global discrete time domain $T$, and maintains an internal table of timers indexed by session identifiers $\mathsf{sid}$. Each entry $t_{\mathsf{sid}}$ either stores a countdown integer or $\perp$. For all $\mathsf{sid}$, initialize $t_{\mathsf{sid}} := \perp$.
\begin{itemize}
    \item Upon receiving $ \langle \mathsf{GetTime}, \mathsf{sid} \rangle $ from a party $P$ or the environment, return the current value of $t_{\mathsf{sid}}$ to party $P$.

    \item Upon receiving $ \langle \mathsf{ResetTime}, \mathsf{sid} \rangle $ from a party $P$, set $t_{\mathsf{sid}} := \perp$ and return $\langle \mathsf{TimeOK}, \mathsf{sid} \rangle $ to party $P$.

    \item Upon receiving $ \langle \mathsf{TimeStart}, \mathsf{sid}, \mathsf{phase}_p, \delta \rangle $ from a party $P$: if $t_{\mathsf{sid}} \neq \perp$, ignore the request. 
    \item Otherwise: 
    \begin{enumerate}
        \item Set $t_{\mathsf{sid}} := \delta$, and return $ \langle \mathsf{TimeOK}, \mathsf{sid} \rangle $ to party $P$.
        \item From this point onward, in each global round, decrement $t_{\mathsf{sid}}$ by $1$ if $t_{\mathsf{sid}} \in \mathbb{N}$.
        \item When $t_{\mathsf{sid}} = 0$, send $ \langle \mathsf{TimeOver}, \mathsf{sid}, \mathsf{phase}_p, \delta \rangle $ to party $P$.
    \end{enumerate}
\end{itemize}
\end{myframebox}
\captionof{figure}{Timer Functionality $\mathcal{G}_{\mathsf{TIME}}$}
\label{fig:func-timer}
\end{center}

\subsection{The Timer Functionality $\mathcal{G}_{\mathsf{TIME}}$}

To ensure protocol termination under network delay attacks, the Tendermint protocol integrates a timeout mechanism formalized as the ideal functionality $\mathcal{G}_{\mathsf{TIME}}$ within the UC framework. This functionality models a logical timer enabling honest parties to detect and respond to extended inactivity across protocol phases, serving as a fundamental building block for progress guarantees under partial synchrony, where message delays are eventually bounded. $\mathcal{G}_{\mathsf{TIME}}$ maintains phase-specific timers for protocol sessions identified by unique session identifiers $\mathsf{sid}$, initially set to $\perp$. Honest parties interact via three interfaces: $\mathsf{GetTime}$ to query the current timer value, $\mathsf{ResetTime}$ to clear it, and $ \langle \mathsf{TimeStart,\, sid,\, phase_p,\, \delta} \rangle $ to initiate a countdown of $\delta$ time units. In compliance with UC scheduling semantics and upon expiration, the invoking party obtains a timeout notification $ \langle \mathsf{TimeOver,\, sid,\, phase_p,\, \delta} \rangle $. This countdown, which cannot be interrupted by an adversary, ensures predictable timeout behavior: when proposals or quorum votes are delayed beyond the timeout, the protocol reliably progresses to the next round, visible to honest parties while remaining concealed from the adversary. The formal specification is provided (see Fig.~\ref{fig:func-timer}), and $\mathcal{G}_{\mathsf{TIME}}$ satisfies UC simulatability requirements, underpinning Tendermint’s termination proof under universal composition.

\subsection{The Proposer Selection Functionality}

\begin{center}
\centering
\begin{myframebox}[title=Proposer Selection Functionality $\mathcal{F}_{\mathsf{GetProposer}}$]
{Description}: Maintains proposer rotation by tracking block $\mathsf{height}$ $h_p$, round number $\mathsf{round}_p$, and historical proposer ($\mathsf{preProposer}$), using the parameter $\mathsf{blocksPerProposer}$. It initializes with $\mathsf{preProposer} = \mathsf{nil}$, $\mathsf{proposerIndex} = 0$, $\mathsf{height} = 0$, and $\mathsf{round} = 0$, while loading $\mathsf{size}$ and $\mathsf{blocksPerProposer}$ from the chain configuration.

\textbf{GetProposer}
\item Upon receiving $\langle \mathsf{GetProposer}, \mathsf{preProposer}, \mathsf{sid}, \mathsf{pid}, h_p, \newline \mathsf{round}_p, |V| \rangle$ from party $P$:
\begin{enumerate}
    \item Set $\mathsf{height} \leftarrow h_p$, $\mathsf{round} \leftarrow \mathsf{round}_p$, $\mathsf{size} \leftarrow |V|$.
    \item Set $\mathsf{proposerOffset} \leftarrow \mathsf{GetIndexByString}(\mathsf{preProposer})$ if $\mathsf{preProposer} \neq \mathsf{nil}$, else $\mathsf{proposerOffset} \leftarrow 0$.
    \item If $\mathsf{height}$ \% $\mathsf{blocksPerProposer} = 0$, increment $\mathsf{proposerOffset}$.
    \item Set $\mathsf{roundOffset} \leftarrow \mathsf{round} \% \mathsf{size}$.
    \item Compute $\mathsf{currentIndex} \leftarrow (\mathsf{proposerOffset} + \mathsf{roundOffset}) \% \mathsf{size}$.
    \item Update $\mathsf{preProposer} \leftarrow \mathsf{GetByIndex}(\mathsf{currentIndex})$.
    \item Send $\langle \mathsf{Proposer}, \mathsf{preProposer}, \mathsf{sid}, \mathsf{pid}, \mathsf{height}, \mathsf{round} \rangle$ to $P$.
\end{enumerate}
\end{myframebox}
\captionof{figure}{Proposer Selection Functionality $\mathcal{F}_{\mathsf{GetProposer}}$}
\label{fig:func-getproposer}
\end{center}

The proposer selection functionality $\mathcal{F}_{\mathsf{GetProposer}}$ formalizes the mechanism of rotating block proposers in a predictable yet fair manner (see Fig.~\ref{fig:func-getproposer}). It maintains state variables for block height, round number, and the most recent proposer, while relying on system parameters such as the total number of validators ($\mathsf{size}$) and the rotation interval ($\mathsf{blocksPerProposer}$). Upon a request, the functionality computes the active proposer index by combining a proposer offset, derived from historical proposer information, with a round offset, determined by the modulo of the current round. This ensures that proposer roles advance deterministically across both block heights and consensus rounds, thereby preventing adversarial manipulation of leader election. 

\subsection{The Authentication Functionality}

\begin{center}
\centering
\begin{myframebox}[title=Authentication Functionality $\mathcal{F}_{\mathsf{AUTH}}$]
\begin{itemize}
    \item Upon receiving $\langle \mathsf{Send}, \mathsf{sid}, B, m \rangle$ from party $A$, send $\langle \mathsf{Sent}, \mathsf{sid}, A, B, m \rangle$ to $\mathcal{A}$.
    \item Upon receiving $\langle \mathsf{Send}, \mathsf{sid}, B', m' \rangle$ from the $\mathcal{A}$, do: If $A$ is corrupted: Output $\langle \mathsf{Sent}, \mathsf{sid}, A, m' \rangle$ to party $B'$. Else: Output $\langle \mathsf{Sent}, \mathsf{sid}, A, m \rangle$ to party $B$. Halt.
    \item Upon receiving $\langle \mathsf{Register}, \mathsf{sid}, A, pk \rangle$ from party $A$, send $\langle \mathsf{RegApply}, \mathsf{sid}, pk, A \rangle$ to $\mathcal{A}$.
    \item Upon receiving $\langle \mathsf{RegStatus}, \mathsf{sid}, A \rangle$ from the $\mathcal{A}$, do: If $\mathsf{RegStatus} == 1$: Send $\langle \mathsf{RegisterSuccess}, \mathsf{sid}, A \rangle$ to party $A$. Else: Output $\langle \mathsf{RegisterFailure}, \mathsf{sid}, A \rangle$ to party $A$. Halt.
    \item Upon receiving $\langle \mathsf{Lookup}, \mathsf{sid}, A, pk \rangle$ from party $A$, send $\langle \mathsf{LookApply}, \mathsf{sid}, pk, A \rangle$ to $\mathcal{A}$.
    \item Upon receiving $\langle \mathsf{LookStatus}, \mathsf{sid}, A \rangle$ from the $\mathcal{A}$, do: If $\mathsf{LookStatus} == 1$: Send $\langle \mathsf{LookupSuccess}, \mathsf{sid}, A \rangle$ to party $A$. Else: Output $\langle \mathsf{LookupFailure}, \mathsf{sid}, A \rangle$ to party $A$. Halt.
    \item Upon receiving $\langle \mathsf{Delete}, \mathsf{sid}, A, pk \rangle$ from party $A$, send $\langle \mathsf{DeleteApply}, \mathsf{sid}, pk, A \rangle$ to $\mathcal{A}$.
    \item Upon receiving $\langle \mathsf{DeleteStatus}, \mathsf{sid}, A \rangle$ from the $\mathcal{A}$, do: If $\mathsf{DeleteStatus} == 1$: Send $\langle \mathsf{DeleteSuccess}, \mathsf{sid}, A \rangle$ to party $A$. Else: Output $\langle \mathsf{DeleteFailure}, \mathsf{sid}, A \rangle$ to party $A$. Halt.
\end{itemize}
\end{myframebox}
\captionof{figure}{Authentication Functionality $\mathcal{F}_{\mathsf{AUTH}}$}
\label{fig:func-auth}
\end{center}

The Tendermint protocol relies on the ideal functionality $\mathcal{F}_{\mathsf{AUTH}}$ (see Fig.~\ref{fig:func-auth}) to guarantee the authenticity and integrity of communication. This functionality follows the universally composable framework introduced by Canetti~\cite{canetti2004universally}, which provides a formal foundation for secure authentication. $\mathcal{F}_{\mathsf{AUTH}}$ is responsible for handling message transmission as well as key management operations, including registration, lookup, and deletion. When a party $A$ initiates a message $\langle \mathsf{Send}, \mathsf{sid}, B, m \rangle$, the functionality ensures its delivery while preserving adversarial influence only under corruption conditions. Similarly, registration, lookup, and deletion requests are processed through $\mathcal{A}$, with the outcome explicitly communicated to the requesting party.

\subsection{The Signature Functionality}

\begin{center}
\centering
\begin{myframebox}[width = 1.02\linewidth,title=Signature Functionality $\mathcal{F}_{\mathsf{SIG}}$]

\textbf{Key Generation}
\begin{itemize}
\item Upon receiving a value $\langle \mathsf{KeyGen}, \mathsf{sid} \rangle$ from some party $\mathsf{S}$, verify that $\mathsf{sid} = \langle \mathsf{S}, \mathsf{sid}' \rangle$ for some $\mathsf{sid}'$. If not, then ignore the request. Else, hand $\langle \mathsf{KeyGen}, \mathsf{sid} \rangle$ to the adversary.  
\item Upon receiving $\langle \mathsf{Verification Key}, \mathsf{sid}, \mathsf{pid}, v \rangle$ from the adversary, output $\langle \mathsf{Verification Key}, \mathsf{sid}, \mathsf{pid}, v \rangle$ to $\mathsf{S}$, and record the pair $\langle \mathsf{S}, v \rangle$.  
\end{itemize}

\textbf{Signature Generation}
\begin{itemize}
\item Upon receiving a value $\langle \mathsf{Sign}, \mathsf{sid}, \mathsf{pid}, m \rangle$ from $\mathsf{S}$, verify that $\mathsf{sid} = \langle \mathsf{S}, \mathsf{sid}' \rangle$ for some $\mathsf{sid}'$. If not, then ignore the request. Else, send $\langle \mathsf{Sign}, \mathsf{sid}, \mathsf{pid}, m \rangle$ to the adversary.  
\item Upon receiving $\langle \mathsf{Signature}, \mathsf{sid}, \mathsf{pid}, m, \sigma \rangle$ from the adversary, verify that no entry $\langle m, \sigma, 0 \rangle$ is recorded. If it is, then output an error message to $\mathsf{S}$ and halt. Else, output $\langle \mathsf{Signature}, \mathsf{sid}, \mathsf{pid}, m, \sigma \rangle$ to $\mathsf{S}$, and record the entry $\langle m, \sigma, 1 \rangle$.  
\end{itemize}

\textbf{Signature Verification}
\begin{itemize}
\item Upon receiving a value $\langle \mathsf{Verify}, \mathsf{sid}, \mathsf{pid}, m, \sigma, v' \rangle$ from party $\mathsf{P}$, hand $\langle \mathsf{Verify}, \mathsf{sid}, \mathsf{pid}, m, \sigma, v' \rangle$ to the adversary.  
Upon receiving $\langle \mathsf{Verified}, \mathsf{sid}, \mathsf{pid}, m, \phi \rangle$ from the adversary, do: 
\begin{enumerate}
\item If $v' = v$ and the entry $\langle m, \sigma, 1 \rangle$ is recorded, then set $f = 1$. (This condition guarantees completeness: If the verification key $v'$ is the registered one and $\sigma$ is a legitimately generated signature for $m$, then the verification succeeds.)
\item Else, if $v' = v$, the signer is not corrupted, and no entry $\langle m, \sigma', 1 \rangle$ for any $\sigma'$ is recorded, then set $f = 0$. (This condition guarantees unforgeability: If $v'$ is the registered one, the signer is not corrupted, and never signed $m$, then the verification fails.)
\item Else, if there is an entry $\langle m, \sigma, f' \rangle$ recorded then let $f = f'$. (This condition guarantees consistency: All verification requests with identical parameters will result in the same answer.)
\item Else, $f = \phi$ and record the entry $\langle m, \sigma', \phi \rangle$.
\end{enumerate}
  Output $\langle \mathsf{Verified}, \mathsf{id}, m, f \rangle$ to $\mathsf{P}$.  
\end{itemize}
\end{myframebox}
\captionof{figure}{Signature Functionality $\mathcal{F}_{\mathsf{SIG}}$}
\label{fig:func-sig}
\end{center}

The signature functionality $\mathcal{F}_{\mathsf{SIG}}$~\cite{canetti2004universally} provides a formal abstraction for digital signature schemes in the UC framework, supporting key generation, signing, and verification (see Fig.~\ref{fig:func-sig}). It interacts with the adversary $\mathcal{A}$ in every stage, thereby modeling the influence of cryptographic delays and adversarial scheduling while preserving fundamental guarantees of authenticity and unforgeability. In particular, verification succeeds only if the verification key matches the registered one and the signature corresponds to a legitimately signed message, ensuring completeness. Conversely, if the signer is uncorrupted and no valid signature has been produced, verification must fail, ensuring unforgeability. Consistency across repeated verification attempts is also enforced, preventing contradictory outcomes. This explicit modeling of signature behavior reflects realistic security assumptions in adversarial environments while still enabling rigorous security proofs.

\subsection{The Broadcast Functionality}

To ensure secure communication under adversarial delay conditions, Tendermint adopts standard UC ideal functionalities. Specifically, the broadcast functionality $\mathcal{F}_{\mathsf{BC}}$ and the authenticated channel $\mathcal{F}_{\mathsf{Ch}}$ are incorporated directly from the formal models of Kiayias et al.~\cite{kiayias2022peredi}, originally designed for privacy-preserving and regulated CBDC infrastructures. The functionality $\mathcal{F}_{\mathsf{BC}}$ enables global dissemination of proposals and votes by accepting $\langle \mathsf{Broadcast}, \mathsf{sid}, m \rangle$ from a party $\mathsf{P}$ and delivering $\langle \mathsf{Broadcasted}, \mathsf{sid}, \mathsf{P}, m \rangle$ to all parties (see Fig.~\ref{fig:func-broadcast}). The formal specification thus provides a reliable abstraction for message diffusion, ensuring that protocol-level guarantees such as agreement and validity can be analyzed within the universal composition framework.

\begin{center}
\centering
\begin{myframebox}[title=Broadcast Functionality $\mathcal{F}_{\mathsf{BC}}$]
Broadcast functionality $\mathcal{F}_{\mathsf{BC}}$ parameterized by the set $\mathbb{M} = \{\mathsf{M}_1, \ldots, \mathsf{M}_D\}$ proceeds as follows:
\begin{itemize}
    \item Upon receiving $\langle \mathsf{Broadcast}, \mathsf{sid}, m \rangle$ from a party $\mathsf{P}$, send $\langle \mathsf{Broadcasted}, \mathsf{sid}, \mathsf{P}, m \rangle$ to all entities in the set $\mathbb{M}$ and to $\mathcal{A}$.
\end{itemize}
\end{myframebox}
\captionof{figure}{Broadcast Functionality $\mathcal{F}_{\mathsf{BC}}$} 
\label{fig:func-broadcast}
\end{center}

\subsection{The Enhanced Communication Channel Functionality}

The enhanced channel functionality $\mathcal{F}_{\mathsf{Ch}}$ ~\cite{kiayias2022peredi} supports flexible sender-receiver mappings with varying anonymity levels (see Fig.~\ref{fig:func-ch}). It generalizes authenticated and anonymous communication by parameterizing the information $\Delta$ revealed to the adversary. Examples include the sender-recipient anonymous channel $\mathcal{F}_{\mathsf{Ch}}^{\mathsf{sra}}$, the sender-sender anonymous channel $\mathcal{F}_{\mathsf{Ch}}^{\mathsf{ssa}}$, and the fully anonymous channel $\mathcal{F}_{\mathsf{Ch}}^{\mathsf{fa}}$, each of which leaks only the message length $\lvert m \rvert$ while hiding identities. By concealing metadata such as sender and receiver while ensuring delivery guarantees, these functionalities mitigate timing correlation and traffic analysis under adversarial delays. Moreover, the specification allows adversarial scheduling of message delivery through explicit acknowledgment, which captures the realistic constraints of synchronous communication under the UC framework. 

\begin{center}
\centering
\begin{myframebox}[title=Enhanced Communication Channel Functionality $\mathcal{F}_{\mathsf{Ch}}$]
Let $\mathcal{P}$ define a set of parties where $\mathcal{S}$ and $\mathcal{R}$ denote two parties of the set as the sender and receiver of a message $m$ respectively. $\Delta$ is defined as follows based on parameters of functionality. Message identifier $\mathsf{mid}$ is selected freshly by the functionality.
\begin{enumerate}
    \item Upon input $\langle \mathsf{Send}, \mathsf{sid}, \mathcal{R}, m \rangle$ from $\mathcal{S}$, output $\langle \mathsf{Send}, \mathsf{sid}, \Delta, \mathsf{mid} \rangle$ to $\mathcal{A}$.
    \item Upon receiving $\langle \mathsf{Ok}, \mathsf{sid}, \mathsf{mid} \rangle$ from $\mathcal{A}$, send $\langle \mathsf{Received}, \mathsf{sid}, \mathcal{S}, m \rangle$ to $\mathcal{R}$.
\end{enumerate}
Set $\Delta$ based on the following parameterized functions:
\begin{itemize}
    \item For $\mathcal{F}^{\mathsf{ac}}_{\mathsf{Ch}}$ set $\Delta = (\mathcal{S}, \mathcal{R}, m)$. Upon receiving $\langle \mathsf{Ok.Snd}, \mathsf{sid}, \mathsf{mid} \rangle$ from $\mathcal{A}$, send $\langle \mathsf{Continue}, \mathsf{sid} \rangle$ to $\mathcal{S}$.\footnote{This gives more power to adversary $\mathcal{A}$ who decides when the sender can proceed as sequential message sending is required in the UC model.}
    \item For $\mathcal{F}^{\mathsf{sra}}_{\mathsf{Ch}}$ set $\Delta = (\mathcal{S}, \lvert m \rvert)$.
    \item For $\mathcal{F}^{\mathsf{ssa}}_{\mathsf{Ch}}$ set $\Delta = (\mathcal{R}, \lvert m \rvert)$.
    \item For $\mathcal{F}^{\mathsf{fa}}_{\mathsf{Ch}}$ set $\Delta = \lvert m \rvert$.
    \item For $\mathcal{F}^{\mathsf{sc}}_{\mathsf{Ch}}$ set $\Delta = (\mathcal{S}, \mathcal{R}, \lvert m \rvert)$. Upon receiving $\langle \mathsf{Ok.Snd}, \mathsf{sid}, \mathsf{mid} \rangle$ from $\mathcal{A}$, send $\langle \mathsf{Continue}, \mathsf{sid} \rangle$ to $\mathcal{S}$.
    \item For $\mathcal{F}^{\mathsf{sa}}_{\mathsf{Ch}}$ set $\Delta = (\mathcal{R}, m)$.
\end{itemize}
\noindent Additional message handling rules:
\begin{enumerate}
    \item[1'.] Upon receiving $\langle \mathsf{Ok}, \mathsf{sid}, \mathsf{mid} \rangle$ from $\mathcal{A}$, send $\langle \mathsf{Received}, \mathsf{sid}, m, \mathsf{mid} \rangle$ to $\mathcal{R}$.
    \item[2'.] Upon receiving $\langle \mathsf{Send}, \mathsf{sid}, \mathsf{mid}, m' \rangle$ from $\mathcal{R}$, output $\langle \mathsf{Send}, \mathsf{sid}, \mathcal{R}, m', \mathsf{mid} \rangle$ to $\mathcal{A}$.
\end{enumerate}
\end{myframebox}
\captionof{figure}{Enhanced Communication Channel Functionality $\mathcal{F}_{\mathsf{Ch}}$} 
\label{fig:func-ch}
\end{center}

\subsection{The Synchronization Functionality}

The ideal synchronization functionality $\mathcal{F}_{\mathsf{SYNC}}$, originally introduced by Katz et al. \cite{achenbach2015synchronous} as $\mathcal{F}_{\mathsf{clock}}$, models a global synchronization service that provides parties with a consistent notion of round progression. In this paper, we refer to it as $\mathcal{F}_{\mathsf{SYNC}}$. Each party can signal the completion of its current round by sending a $\mathsf{RoundOK}$ message. Once all honest parties have signaled, the functionality resets the internal round indicators, thereby marking the transition to the next round. This ensures that progress is only made when all honest participants are synchronized. Furthermore, parties can issue a $\mathsf{RequestRound}$ query to check whether the round has advanced, allowing them to coordinate their actions consistently. At the same time, the adversary $\mathcal{A}$ is notified of each party’s $\mathsf{RoundOK}$ signal, reflecting that round synchronization events are observable in practice. The formal specification is provided (see Fig.~\ref{fig:func-sync}), and $\mathcal{F}_{\mathsf{SYNC}}$ satisfies UC simulatability requirements, underpinning the modeling of synchronous progress in consensus protocols under universal composition.

\begin{center}
\centering
\begin{myframebox}[title=Synchronization Functionality $\mathcal{F}_{\mathsf{SYNC}}$]
Initialize for each party $p_i$ a bit $d_i := 0$. 
\begin{itemize}
    \item Upon receiving message $\langle \mathsf{RoundOK} \rangle$ from party $p_i$ set $d_i = 1$. If for all honest parties $d_i = 1$, then reset all $d_i$ to 0. In any case, send $\langle \mathsf{switch, p_i} \rangle$ to $\mathcal{A}$.
    \item Upon receiving message $\langle \mathsf{RequestRound} \rangle$ from $p_i$, send $d_i$ to $p_i$.
\end{itemize}
\end{myframebox}
\captionof{figure}{Synchronization Functionality $\mathcal{F}_{\mathsf{SYNC}}$}
\label{fig:func-sync}
\end{center}

\subsection{The Write-Ahead Log Recovery Functionality}

The write-ahead log recovery functionality $\mathcal{F}_{\mathsf{ReplyWAL}}$ models the mechanism by which Tendermint nodes restore their consensus state after failures or restarts (see Fig.~\ref{fig:func-replywal}). It formalizes the sequence of operations required to guarantee consistency between persistent WAL data and the current blockchain height. The functionality first checks the system’s WAL mode, distinguishing between writing and non-writing configurations. It then attempts to construct an iterator over stored records; if this process fails, recovery is aborted. Once a valid iterator is obtained, the last entry is extracted and its height compared against the current chain height. This ensures that the system neither replays outdated states nor advances beyond what the chain already commits. When replaying, proposal entries are deserialized and processed, enabling the node to rejoin the consensus phase seamlessly. Finally, the functionality communicates success or error conditions back to the requesting party or environment, ensuring deterministic recovery outcomes. 

\begin{center}
\centering
\begin{myframebox}[title=Functionality $\mathcal{F}_{\mathsf{ReplyWAL}}$]

\textbf{ReplyWAL}
\\
Upon receiving $\langle \mathsf{ReplyWal}, \mathsf{sid}, \mathsf{pid}, \mathsf{phase}_p \rangle$ from party $\mathsf{P}$:

\begin{enumerate}
    \item {Initial Check}:
      \begin{itemize}
          \item Obtain the current chain height $\mathsf{currentHeight}$.
          \item If $\mathsf{walWriteMode}$ is NonWalWrite, send $\langle \mathsf{WalRestored}, \mathsf{sid}, \mathsf{pid}, \mathsf{no\_write\_mode}, \newline \mathsf{currentHeight} \rangle$ to $\mathsf{P}$ and return.
      \end{itemize}
    \item {Iterator Operations}:
      \begin{itemize}
          \item Create a WAL log iterator $\mathsf{iterator}$.
          \item If $\mathsf{iterator}$ creation fails, send $\langle \mathsf{WalRestored}, \mathsf{sid}, \mathsf{pid}, \mathsf{iterator\_fail}, \newline \mathsf{currentHeight} \rangle$ to $\mathsf{P}$ and return.
          \item Skip to the last record and check if there is a previous record.
          \item If there is a previous record, get the last record data $\mathsf{lastData}$ and deserialize it into $\mathsf{lastEntry}$.
      \end{itemize}
    \item {Height Consistency Verification}:
      \begin{itemize}
          \item Record the height $\mathsf{Height} = \mathsf{lastEntry.Height}$.
          \item If $\mathsf{currentHeight} + 1 < \mathsf{lastEntry.Height}$, send $\langle \mathsf{ERROR}, \mathsf{sid}, \mathsf{pid}, \mathsf{wal\_height\_inconsistent} \rangle$ to $\mathsf{P}$ and return.
          \item If $\mathsf{currentHeight} \geq \mathsf{lastEntry.Height}$, send $\langle \mathsf{WalRestored}, \mathsf{sid}, \mathsf{pid}, \mathsf{chain\_ahead}, \newline \mathsf{currentHeight} \rangle$ to $\mathsf{P}$ and return.
      \end{itemize}
    \item {Replay Execution}:
      \begin{itemize}
          \item Process the WAL record type based on $\mathsf{lastEntry.Type}$:
            \begin{itemize}
                \item If it is $\mathsf{PROPOSAL\_ENTRY}$:
                  \begin{itemize}
                      \item Deserialize the proposal and send $\langle \mathsf{EnterPrecommit}, \mathsf{sid}, \mathsf{pid}, \mathsf{proposal} \rangle$ to $\mathsf{P}$.
                      \item Otherwise, send $\langle \mathsf{ERROR}, \mathsf{sid}, \mathsf{pid}, \mathsf{replay\_fail} \rangle$ to $\mathsf{P}$ and return.
                  \end{itemize}
                \item Otherwise, log a warning and send $\langle \mathsf{ERROR}, \mathsf{sid}, \mathsf{pid}, \mathsf{invalid\_entry\_type} \rangle$ to $\mathsf{P}$.
            \end{itemize}
      \end{itemize}
    \item {Final Return}:
      \begin{itemize}
          \item Send $\langle \mathsf{WalRestored}, \mathsf{sid}, \mathsf{pid}, \mathsf{ReplySuccess}, \newline \mathsf{lastEntry.Height} \rangle$ to $\mathsf{P}$.
      \end{itemize}
\end{enumerate}

\textbf{SetMode}
\\
Upon receiving $\langle \mathsf{SetMod}, \mathsf{sid}, \mathsf{pid}, \mathsf{mode} \rangle$ from environment $\mathcal{E}$:

\begin{itemize}
    \item If $\mathsf{mode}$ is in $\{\mathsf{WalWrite}, \mathsf{NonWalWrite}\}$, set $\mathsf{walWriteMode}$ to mode and send $\langle \mathsf{ModeSet}, \mathsf{sid}, \mathsf{pid}, \mathsf{mode} \rangle$ to $\mathcal{E}$.
    \item Otherwise, send $\langle \mathsf{ERROR}, \mathsf{sid}, \mathsf{pid}, \mathsf{invalid\_mode} \rangle$ to $\mathcal{E}$.
\end{itemize}

\end{myframebox}
\captionof{figure}{Write-Ahead Log Recovery Functionality $\mathcal{F}_{\mathsf{ReplyWAL}}$}
\label{fig:func-replywal}
\end{center}

\subsection{The Functionality $\mathcal{F}_{\mathsf{Tendermint}}$}

The ideal functionality $\mathcal{F}_{\mathsf{Tendermint}}^{V,\Delta,\delta,\tau}[\mathcal{F}_{\mathsf{AUTH}}, \mathcal{F}_{\mathsf{BC}}, \mathcal{F}_{\mathsf{SIG}}, \mathcal{F}_{\mathsf{SYNC}}, \mathcal{F}_{\mathsf{TIME}}, \mathcal{F}_{\mathsf{ReplyWAL}}]$ formally specifies the termination properties of the Tendermint consensus protocol under network delay attacks within the Universally Composable (UC) framework. Parameterized by the validator set $V$, network one-way delay $\Delta$, normal execution time $\delta$, and adaptive per-phase timeouts $\tau$, it integrates timing and cryptographic primitives to orchestrate consensus via phased voting, locking, and adaptive timeouts, thereby ensuring termination despite adversarial delays. The functionality coordinates six auxiliary services—authenticated communication, reliable broadcast, signature verification, synchronization, timeout management, and write-ahead logging—which collectively enable secure messaging, adaptive leader rotation, cryptographic operations, and fault-tolerant state recovery. These services include authentication ($\mathcal{F}_{\mathsf{AUTH}}$) and signature ($\mathcal{F}_{\mathsf{SIG}}$)~\cite{canetti2004universally}, broadcast ($\mathcal{F}_{\mathsf{BC}}$) and enhanced channels ($\mathcal{F}_{\mathsf{Ch}}$)~\cite{kiayias2022peredi}, proposer selection ($\mathcal{F}_{\mathsf{GetProposer}}$), synchronization ($\mathcal{F}_{\mathsf{SYNC}}$)~\cite{achenbach2015synchronous}, and write-ahead logging ($\mathcal{F}_{\mathsf{ReplyWAL}}$).

\begin{table}[htbp]
  \caption{Parameters of $\mathcal{F}_{\mathsf{Tendermint}}$}
  \begin{center}
    \scriptsize
    \begin{tabular}{|l|l|}
      \hline
        \multicolumn{1}{|c|}{\textbf{Parameter}}
        & \multicolumn{1}{c|}{\textbf{Description}} \\
      \hline
      $\mathrm{V}$ & Validator set. \\ \hline
      $\Delta$ & Network one-way delay. \\ \hline
      $\mathcal{F}_{\mathsf{AUTH}}$ & Ideal functionality for authentication. \\ \hline
      $\mathcal{F}_{\mathsf{BC}}$ & Ideal functionality for broadcast. \\ \hline
      $\mathcal{F}_{\mathsf{SIG}}$ & Ideal functionality for signature. \\ \hline
      $\mathcal{F}_{\mathsf{SYNC}}$ & Ideal functionality for synchronization. \\ \hline
      $\mathcal{G}_{\mathsf{TIME}}$ & Ideal functionality for timing. \\ \hline
      $\mathcal{F}_{\mathsf{ReplyWAL}}$ & Ideal functionality for write‑ahead log. \\ \hline
    \end{tabular}
    \normalsize
    \label{tab:params}
  \end{center}
\end{table}

\begin{table}[htbp]
  \caption{Symbol Explanation of $\mathcal{F}_{\mathsf{Tendermint}}$}
  \begin{center}
    \scriptsize
    \begin{tabularx}{\columnwidth}{|l|X|}
      \hline
      \multicolumn{1}{|c|}{\textbf{Symbol}} &
      \multicolumn{1}{c|}{\textbf{Explanation}} \\
      \hline
      $|V|$ & Total number of validators. \\ \hline
      $\delta$ & Normal protocol execution time, provided by the configuration file. \\ \hline
      $\sigma$ & Adversary’s attack delay, specified by $\mathcal{S}$. \\ \hline
      $\tau$ & Phase timeout duration, provided by the configuration file. \\ \hline
      $B$ & Threshold for blocks per proposer. \\ \hline
      $h_{p}$ & Current consensus height of node $p$. \\ \hline
      $\mathrm{round}_{p}$ & Current round number of node $p$. \\ \hline
      $\mathrm{phase}_{p}$ & Current phase of node $p$ (propose, prevote, precommit, commit). \\ \hline
      $\mathrm{count}_{\mathrm{phase}_p}$ & Number of votes collected in the current phase. \\ \hline
      $\mathrm{lockedValue}_p$ & Value locked by node $p$ in the current round. \\ \hline
      $\mathrm{lockedRound}_p$ & Round in which node $p$’s value is locked. \\ \hline
      $\mathrm{validValue}_p$ & Valid value for node $p$ in the current round. \\ \hline
      $\mathrm{validRound}_p$ & Round in which node $p$’s value is valid. \\ \hline
      $\mathrm{decision}_p[]$ & Final consensus values for each height. \\ \hline
      $\mathrm{preProposer}$ & Index of the proposer in the previous round. \\ \hline
      $\mathrm{txID}_{\mathrm{random}}$ & Set of random transaction IDs. \\ \hline
      $\mathrm{count}_{\mathrm{random}}(\mathrm{txID})$ & Count of random-exclusion votes for a given transaction ID. \\ \hline
      $\ast$ & Arbitrary parameters. \\ \hline
    \end{tabularx}
    \normalsize
    \label{tab:symbols}
  \end{center}
\end{table}

To mitigate delay attacks, $\mathcal{F}_{\mathsf{Tendermint}}$ executes a five-phase pipeline—NewRound, Propose, Prevote, PreCommit, and Commit—while integrating dynamic timeouts $\tau^{r}_{\mathsf{phase}} = \tau^{\mathsf{init}}_{\mathsf{phase}} + r \cdot \tau^{\mathsf{step}}_{\mathsf{phase}}$. It begins with $\langle \text{NEWHEIGHT} \rangle$ and $\langle \text{NEWROUND} \rangle$ messages, signaling the start of a new cycle with WAL recovery. In the Propose phase, the designated proposer (via $\mathcal{F}_{\mathsf{GetProposer}}$) broadcasts its block only if $\delta + \sigma \leq \tau$. The Prevote and PreCommit phases involve validators hashing transactions, casting votes (including nil-votes), and locking values on a supermajority, advancing rounds based on timeout conditions. Finally, the Commit phase finalizes the block once more than $2f+1$ precommits are collected, after which $\mathcal{F}_{\mathsf{SYNC}}$ advances the consensus height.For each validator $p \in V$, the functionality tracks the current phase, height, round, vote counters, decision history, locked values, validation state, and proposer index, while also managing block thresholds, randomness, and WAL-based recovery. The parameters and symbols involved are listed in Table~\ref{tab:params} and Table~\ref{tab:symbols}, where their roles are explained in detail.

\begin{center}
\begin{myframebox}[title=The Functionality $\mathcal{F}_{\mathsf{Tendermint}}$]
\textbf{Corrupt}\\
Upon receiving a message $\langle \mathsf{Corrupt}, \mathsf{sid},\mathsf{pid},  \mathsf{Validator}\rangle$:
\begin{itemize}
    \item If $\mathsf{Validator}_{pid} \in \{Validator_{{pid}_1}, . . . , Validator_{{pid}_n}\}$, then record $Validator_{pid}$ as corrupted.
\end{itemize}
\textbf{NewHeight}\\
Upon receiving message $\langle \mathsf{NEWHEIGHT}, h_p, \mathsf{round}_p,  \ast \rangle$ from $\mathcal{S}$, while $\mathsf{phase}_p = \mathsf{propose}$:
\begin{itemize}
    \item $h_p = h_p +1$, $\mathsf{round}_p = 0$
    \item Send $\langle \mathsf{NEWROUND}, h_p, \mathsf{round}_p, \ast \rangle$ to $\mathcal{S}$.
\end{itemize}
\textbf{NewRound and Proposal}\\
Upon receiving message $\langle \mathsf{NEWROUND}, h_p, \mathsf{round}_p,  \ast \rangle$ from $\mathcal{S}$, while $\mathsf{phase}_p = \mathsf{propose}$:
\begin{itemize}
    \item If $\ast =\langle \mathsf{EnterPreCommit}, \mathsf{sid},\mathsf{pid}, \mathsf{PROPOSAL}_{wal} \rangle$,
    \begin{itemize}
        \item Broadcast $\langle \mathsf{PROPOSAL}_{wal}, h_p, \mathsf{round}_p, v_{wal} \rangle$.
        \item Set $\mathsf{phase}_p = \mathsf{precommit}$.
    \end{itemize}
    \item Otherwise: Send $\langle \mathsf{ReplyWal}, \mathsf{sid},\mathsf{pid}, \mathsf{phase}_p \rangle$ to $\mathcal{F}_{\mathsf{ReplyWAL}}$, and wait for the response from $\mathcal{F}_{\mathsf{ReplyWAL}}$.
    \item If the response is $\langle \mathsf{EnterPreCommit}, \mathsf{sid},\mathsf{pid}, \mathsf{PROPOSAL}_{wal} \rangle$, set $\ast \gets \langle \mathsf{EnterPreCommit}, \mathsf{sid},\mathsf{pid}, \mathsf{PROPOSAL}_{wal} \rangle$
    \item If the response is $\langle \mathsf{WalRestored}, \mathsf{sid},\mathsf{pid}, \mathsf{status}, \mathsf{WALHeight} \rangle$
    \begin{itemize}
        \item If $\mathsf{status} \in \{\mathsf{no\_write\_mode ,iterator\_fail,chain\_ahead}\}$, Send $\langle \mathsf{NEWHEIGHT}, \mathsf{WALHeight}, \mathsf{round}_p, \ast \rangle$ to $\mathcal{S}$.
        \item If $\mathsf{status} = \mathsf{ReplySuccess}$, Send $\langle \mathsf{NEWROUND}, \mathsf{lastEntry.Height}, 0, \ast \rangle$ to $\mathcal{S}$.
    \end{itemize}
    \item If the response is $\langle \mathsf{ERROR, sid, error\_status} \rangle$, Log the error and take appropriate action as needed.
    \item Otherwise: Send $\langle \mathsf{Sleep}, \mathsf{sid},\mathsf{pid}, \mathsf{phase}_p \rangle$ to $\mathcal{S}$ and wait for a response of the form $\langle \mathsf{Wake}, \mathsf{sid},\mathsf{pid}, \mathsf{phase}_p, \sigma \rangle$.
    \item If $\delta + \sigma > \tau^r_{Propose}$:
    \begin{itemize}
        \item Return to previous step.
    \end{itemize}
    \item Otherwise:
    \begin{itemize}
        \item Send $\langle \mathsf{TimeStart}, \mathsf{sid},\mathsf{pid}, \mathsf{phase}_p, \sigma \rangle$ to $\mathcal{G}_{\mathsf{TIME}}$, and suspend execution.
        \item Upon receiving $\langle \mathsf{TimeOver}, \mathsf{sid},\mathsf{pid}, \mathsf{phase}_p, \sigma \rangle$ from $\mathcal{G}_{\mathsf{TIME}}$, resume execution.
        \item Send $\langle \mathsf{TimeStart}, \mathsf{sid},\mathsf{pid}, \mathsf{phase}_p, \delta \rangle$ to $\mathcal{G}_{\mathsf{TIME}}$.  
        \item Send $\langle \mathsf{GetProposer}, \mathsf{preProposer}, \mathsf{sid},\mathsf{pid},h_p, \mathsf{round}_p ,|V|\rangle$ to $\mathcal{F}_{\mathsf{GetProposer}}$ and wait for a response of the form $\langle \mathsf{Proposer}, \mathsf{preProposer}, \mathsf{sid},\mathsf{pid},h_p, \mathsf{round}_p \rangle$.
        \item Send $\langle \mathsf{StartPropose}, \mathsf{sid},\mathsf{pid}, \mathsf{phase}_p \rangle$ to $\mathsf{Proposer}(h_p, \mathsf{round}_p)$ and wait for a response of the form $\langle \mathsf{PROPOSAL}, \mathsf{sid},\mathsf{pid}, \mathsf{phase}_p, v \rangle$.
        \item Send $\langle \mathsf{Verify}, \mathsf{sid},\mathsf{pid}, m(PROPOSAL, h_p, round_p, v), \mathsf{sig},\mathsf{pk} \rangle$ to $\mathcal{F}_{\mathsf{SIG}}$ to verify the message signature and obtain $\mathsf{result}_{sig}$.
        \item Send $\langle \mathsf{Verify}, \mathsf{sid},\mathsf{pid}, \mathsf{Proposer}(h_p, \mathsf{round}_p), v \rangle$ to $\mathcal{F}_{\mathsf{AUTH}}$ to verify the Proposer's identity and obtain $\mathsf{result}_{auth}$.
        \item If $\mathsf{result}_{sig} =0\lor \mathsf{result}_{auth} = 0$:
        \begin{itemize}
            \item Remove $\mathsf{Proposer}(h_p, \mathsf{round}_p)$ from Validator set.
            \item Send $\langle \mathsf{NEWROUND}, h_p, \mathsf{round}_p+1, \ast \rangle$ to $\mathcal{S}$.
        \end{itemize}
        \item If $\mathsf{Proposer}(h_p, \mathsf{round}_p)$ is corrupted,
        \begin{itemize}
            \item Send $\langle \mathsf{LeakValue}, \mathsf{sid},\mathsf{pid}, h_p, \mathsf{round}_p, v \rangle$ to $\mathcal{S}$.
        \end{itemize}
        \item Else if $\mathsf{valid}(v)$ and no $\langle \mathsf{TimeOver}, \mathsf{sid},\mathsf{pid}, \mathsf{phase}_p, \delta \rangle$ has been received from $\mathcal{G}_{\mathsf{TIME}}$:
        \begin{itemize}
            \item Broadcast $\langle \mathsf{PROPOSAL}, h_p, \mathsf{round}_p, v \rangle$.
            \item Update $\mathsf{phase}_p \gets \mathsf{prevote}$.
        \end{itemize}
        \item Otherwise:
        \begin{itemize}
            \item Send $\langle \mathsf{NEWROUND}, h_p, \mathsf{round}_p+1, \ast \rangle$ to $\mathcal{S}$.
        \end{itemize}
    \end{itemize}
\end{itemize}
\textbf{Prevote}\\
Upon receiving message $\langle \mathsf{PROPOSAL}, h_p, \mathsf{round}_p,  v,-1 \rangle$ from $\mathsf{Proposer}(h_p, \mathsf{round}_p)$, while $\mathsf{phase}_p = \mathsf{prevote}$:
\begin{itemize}
    \item Send $\langle \mathsf{Sleep}, \mathsf{sid},\mathsf{pid}, \mathsf{phase}_p \rangle$ to $\mathcal{S}$ and wait for a response of the form $\langle \mathsf{Wake}, \mathsf{sid},\mathsf{pid}, \mathsf{phase}_p, \sigma \rangle$.
    \item If $\delta + \sigma > \tau^r_{Prevote}$:
    \begin{itemize}
        \item Broadcast $\langle \mathsf{PREVOTE}, h_p, \mathsf{round}_p, \mathsf{nil} \rangle$.
    \end{itemize}
    \item Otherwise:
    \begin{itemize}
        \item Send $\langle \mathsf{TimeStart}, \mathsf{sid},\mathsf{pid}, \mathsf{phase}_p, \sigma \rangle$ to $\mathcal{G}_{\mathsf{TIME}}$, and suspend execution.
        \item Upon receiving $\langle \mathsf{TimeOver}, \mathsf{sid},\mathsf{pid}, \mathsf{phase}_p, \sigma \rangle$ from $\mathcal{G}_{\mathsf{TIME}}$, resume execution.
        \item Send $\langle \mathsf{TimeStart}, \mathsf{sid},\mathsf{pid}, \mathsf{phase}_p, \delta \rangle$ to $\mathcal{G}_{\mathsf{TIME}}$.
        \item Send $\langle \mathsf{Verify}, \mathsf{sid},\mathsf{pid}, m( \mathsf{PROPOSAL}, h_p, \mathsf{round}_p, v ), \mathsf{sig} ,\mathsf{pk} \rangle$ to $\mathcal{F}_{\mathsf{SIG}}$ to verify the message signature and obtain $\mathsf{result}_{sig}$.
        \item Send $\langle \mathsf{Verify}, \mathsf{sid},\mathsf{pid}, \mathsf{Proposer}(h_p, \mathsf{round}_p), v \rangle$ to $\mathcal{F}_{\mathsf{AUTH}}$ to verify the Proposer's identity and obtain $\mathsf{result}_{auth}$.
        \item If $\mathsf{result}_{sig} =0\lor \mathsf{result}_{auth} = 0$:
        \begin{itemize}
            \item Remove $\mathsf{Proposer}(h_p, \mathsf{round}_p)$ from Validator set.
        \end{itemize}
        \item Else if $\mathsf{valid}(v) \land (\mathsf{lockedRound}_p =-1 \lor \mathsf{lockedValue}_p = v)$ and no $\langle \mathsf{TimeOver}, \mathsf{sid},\mathsf{pid}, \mathsf{phase}_p, \delta \rangle$ has been received from $\mathcal{G}_{\mathsf{TIME}}$:
        \begin{itemize}
            \item Broadcast $\langle \mathsf{Execute}, v.\mathsf{Transactions} \rangle$ and wait for a response $\langle \mathsf{ReadWriteHash}, H_{\mathsf{exec}} \rangle$.
            \item If $H_{\mathsf{exec}} \neq v.H_{\mathsf{readWrite}}$:
            \begin{itemize}
                \item Broadcast $\langle \mathsf{IdentifyRandom}, v.\mathsf{Transactions} \rangle$ and receive $\mathsf{txID}_{\mathsf{random}}$.
                \item Broadcast $\langle \mathsf{PREVOTE}, h_p, \mathsf{round}_p, \mathsf{nil}, \mathsf{txID}_{\mathsf{random}} \rangle$.
            \end{itemize}
            \item Otherwise:
            \begin{itemize}
                \item Broadcast $\langle \mathsf{PREVOTE}, h_p, \mathsf{round}_p, \mathsf{id}(v) \rangle$.
            \end{itemize}
        \end{itemize}
        \item Otherwise:
        \begin{itemize}
            \item Broadcast $\langle \mathsf{PREVOTE}, h_p, \mathsf{round}_p, \mathsf{nil} \rangle$.
        \end{itemize}
    \end{itemize}
    \item Update $\mathsf{phase}_p \gets \mathsf{precommit}$.
\end{itemize}

Upon receiving message $\langle \mathsf{PROPOSAL}, h_p, \mathsf{round}_p, v, \mathsf{validRound}_p \rangle$ from $\mathsf{Proposer}(h_p, \mathsf{round}_p)$ AND $\mathsf{count}_{\mathsf{prevote}} > 2f+1$ while $\mathsf{phase}_p = \mathsf{propose} \land ( \mathsf{validRound}_p \geq 0 \land \mathsf{validRound}_p < \mathsf{round}_p)$:

\begin{itemize}
    \item Send $\langle \mathsf{Sleep}, \mathsf{sid},\mathsf{pid}, \mathsf{phase}_p \rangle$ to $\mathcal{S}$ and wait for a response of the form $\langle \mathsf{Wake}, \mathsf{sid},\mathsf{pid}, \mathsf{phase}_p, \sigma \rangle$.
    \item If $\delta + \sigma > \tau^r_{Prevote}$:
    \begin{itemize}
        \item Broadcast $\langle \mathsf{PREVOTE}, h_p, \mathsf{round}_p, \mathsf{nil} \rangle$.
    \end{itemize}
    \item Otherwise:
    \begin{itemize}
        \item Send $\langle \mathsf{TimeStart}, \mathsf{sid},\mathsf{pid}, \mathsf{phase}_p, \sigma \rangle$ to $\mathcal{G}_{\mathsf{TIME}}$, and suspend execution.
        \item Upon receiving $\langle \mathsf{TimeOver}, \mathsf{sid},\mathsf{pid}, \mathsf{phase}_p, \sigma \rangle$ from $\mathcal{G}_{\mathsf{TIME}}$, resume execution.
        \item Send $\langle \mathsf{TimeStart}, \mathsf{sid},\mathsf{pid}, \mathsf{phase}_p, \delta \rangle$ to $\mathcal{G}_{\mathsf{TIME}}$.
        \item Send $\langle \mathsf{Verify}, \mathsf{sid},\mathsf{pid}, m( \mathsf{PROPOSAL}, h_p, \mathsf{round}_p, v ), \mathsf{sig} ,\mathsf{pk} \rangle$ to $\mathcal{F}_{\mathsf{SIG}}$ to verify the message signature and obtain $\mathsf{result}_{sig}$.
        \item Send $\langle \mathsf{Verify}, \mathsf{sid},\mathsf{pid}, \mathsf{Proposer}(h_p, \mathsf{round}_p), v \rangle$ to $\mathcal{F}_{\mathsf{AUTH}}$ to verify the Proposer's identity and obtain $\mathsf{result}_{auth}$.
        \item If $\mathsf{result}_{sig} =0\lor \mathsf{result}_{auth} = 0$:
        \begin{itemize}
            \item Remove $\mathsf{Proposer}(h_p, \mathsf{round}_p)$ from Validator set.
        \end{itemize}
        \item Else if $\mathsf{valid}(v) \land (\mathsf{lockedRound}_p \leq \mathsf{validRound}_p \lor \mathsf{lockedValue}_p = v)$ and no $\langle \mathsf{TimeOver}, \mathsf{sid}, \mathsf{pid}, \mathsf{phase}_p, \delta \rangle$ has been received from $\mathcal{G}_{\mathsf{TIME}}$:
        \begin{itemize}
            \item Broadcast $\langle \mathsf{Execute}, v.\mathsf{Transactions} \rangle$ and wait for a response $\langle \mathsf{ReadWriteHash}, H_{\mathsf{exec}} \rangle$.
            \item If $H_{\mathsf{exec}} \neq v.H_{\mathsf{readWrite}}$:
            \begin{itemize}
                \item Broadcast $\langle \mathsf{IdentifyRandom}, v.\mathsf{Transactions} \rangle$ and receive $\mathsf{txID}_{\mathsf{random}}$.
                \item Broadcast $\langle \mathsf{PREVOTE}, h_p, \mathsf{round}_p, \mathsf{nil}, \mathsf{txID}_{\mathsf{random}} \rangle$.
            \end{itemize}
            \item Otherwise:
            \begin{itemize}
                \item Broadcast $\langle \mathsf{PREVOTE}, h_p, \mathsf{round}_p, \mathsf{id}(v) \rangle$.
            \end{itemize}
        \end{itemize}
        \item Otherwise:
        \begin{itemize}
            \item Broadcast $\langle \mathsf{PREVOTE}, h_p, \mathsf{round}_p, \mathsf{nil} \rangle$.
        \end{itemize}
    \end{itemize}
    \item Update $\mathsf{phase}_p \gets \mathsf{precommit}$.
\end{itemize}
Upon receiving message $\langle \mathsf{PREVOTE}, h_p, \mathsf{round}_p, \mathsf{nil},  \mathsf{txID}_{\mathsf{random}} \rangle$ from $\mathsf{Validator}(h_p, \mathsf{round}_p)$, while $\mathsf{phase}_p = \mathsf{prevote}$:
\begin{itemize}
    \item For each $\mathsf{txID} \in \mathsf{txID}_{\mathsf{random}}$, Set $\mathsf{count}_{\mathsf{random}}(\mathsf{txID}) \gets \mathsf{count}_{\mathsf{random}}(\mathsf{txID}) + 1$.
    \item If there exists a $\mathsf{txID}$ such that $\mathsf{count}_{\mathsf{random}}(\mathsf{txID}) \geq f + 1$:
    \begin{itemize}
        \item Broadcast $\langle \mathsf{RemoveTx}, \mathsf{txID} \rangle$ to remove the transaction from the transaction pool.
        \item Reset $\mathsf{count}_{\mathsf{random}}(\mathsf{txID}) \gets 0$.
    \end{itemize}
\end{itemize}
\textbf{PreCommit}\\
Upon receiving message $\langle \mathsf{PREVOTE}, h_p, \mathsf{round}_p,  \mathsf{id}(v) \rangle$ from $\mathsf{Validator}(h_p, \mathsf{round}_p)$, while $\mathsf{phase}_p = \mathsf{precommit}$:
\begin{itemize}
    \item Set $\mathsf{count}_{\mathsf{prevote}} \gets \mathsf{count}_{\mathsf{prevote}} + 1$.
    \item Send $\langle \mathsf{Sleep}, \mathsf{sid},\mathsf{pid}, \mathsf{phase}_p \rangle$ to $\mathcal{S}$ and wait for a response of the form $\langle \mathsf{Wake}, \mathsf{sid},\mathsf{pid}, \mathsf{phase}_p, \sigma \rangle$.
    \item If $\delta + \sigma > \tau^r_{PreCommit}$:
    \begin{itemize}
        \item Broadcast $\langle \mathsf{PRECOMMIT}, h_p, \mathsf{round}_p, \mathsf{nil} \rangle$.
    \end{itemize}
    \item Otherwise:
    \begin{itemize}
        \item Send $\langle \mathsf{TimeStart}, \mathsf{sid},\mathsf{pid}, \mathsf{phase}_p, \sigma \rangle$ to $\mathcal{G}_{\mathsf{TIME}}$, and suspend execution.
        \item Upon receiving $\langle \mathsf{TimeOver}, \mathsf{sid},\mathsf{pid}, \mathsf{phase}_p, \sigma \rangle$ from $\mathcal{G}_{\mathsf{TIME}}$, resume execution.
        \item Send $\langle \mathsf{TimeStart}, \mathsf{sid},\mathsf{pid}, \mathsf{phase}_p, \delta \rangle$ to $\mathcal{G}_{\mathsf{TIME}}$.
        \item Send $\langle \mathsf{Verify}, \mathsf{sid},\mathsf{pid}, m( \mathsf{PREVOTE}, h_p, \mathsf{round}_p, \mathsf{id}(v) ), \mathsf{sig} ,\mathsf{pk} \rangle$ to $\mathcal{F}_{\mathsf{SIG}}$ to verify the message signature and obtain $\mathsf{result}_{sig}$.
        \item Send $\langle \mathsf{Verify}, \mathsf{sid},\mathsf{pid}, \mathsf{Validator}(h_p, \mathsf{round}_p), v \rangle$ to $\mathcal{F}_{\mathsf{AUTH}}$ to verify the Proposer's identity and obtain $\mathsf{result}_{auth}$.
        \item If $\mathsf{result}_{sig} =0\lor \mathsf{result}_{auth} = 0$:
        \begin{itemize}
            \item Remove $\mathsf{Validator}(h_p, \mathsf{round}_p)$ from Validator set.
        \end{itemize}
        \item Else if $\mathsf{valid}(v) \land (\mathsf{count}_{\mathsf{prevote}} > 2f+1)$ and no $\langle \mathsf{TimeOver}, \mathsf{sid},\mathsf{pid}, \mathsf{phase}_p , \delta\rangle$ has been received from $\mathcal{G}_{\mathsf{TIME}}$:
        \begin{itemize}
            \item Set $\mathsf{lockedValue}_p \gets v$, $\mathsf{lockedRound}_p \gets \mathsf{round}_p$.
            \item Broadcast $\langle \mathsf{PRECOMMIT}, h_p, \mathsf{round}_p, \mathsf{id}(v) \rangle$.
        \end{itemize}
        \item Otherwise:
        \begin{itemize}
            \item Broadcast $\langle \mathsf{PRECOMMIT}, h_p, \mathsf{round}_p, \mathsf{nil} \rangle$.
        \end{itemize}
    \end{itemize}
    \item Update $\mathsf{phase}_p \gets \mathsf{commit}$.
\end{itemize}
\textbf{Commit} \\
Upon receiving message $\langle \mathsf{PRECOMMIT}, h_p, \mathsf{round}_p,  \mathsf{id}(v) \rangle$ from $\mathsf{Validator}(h_p, \mathsf{round}_p)$, while $\mathsf{phase}_p = \mathsf{commit}$:
\begin{itemize}
    \item Set $\mathsf{count}_{\mathsf{precommit}} \gets \mathsf{count}_{\mathsf{precommit}} + 1$.
    \item Send $\langle \mathsf{Sleep}, \mathsf{sid},\mathsf{pid}, \mathsf{phase}_p \rangle$ to $\mathcal{S}$ and wait for a response of the form $\langle \mathsf{Wake}, \mathsf{sid},\mathsf{pid}, \mathsf{phase}_p, \sigma \rangle$.
    \item If $\delta + \sigma > \tau^r_{Commit}$:
    \begin{itemize}
        \item Update $\mathsf{phase}_p \gets \mathsf{propose}$ and $\mathsf{round}_q \gets \mathsf{round}_p+1$.
        \item Set $\delta_{\mathsf{round}_{q}} = \delta_{\mathsf{round}_p} + \mathsf{round}_{q} \ast \Delta$
        \item Send $\langle \mathsf{NEWROUND}, h_p, \mathsf{round}_p, \ast \rangle$ to $\mathcal{S}$.
    \end{itemize}
    \item Otherwise:
    \begin{itemize}
        \item Send $\langle \mathsf{TimeStart}, \mathsf{sid},\mathsf{pid}, \mathsf{phase}_p, \sigma \rangle$ to $\mathcal{G}_{\mathsf{TIME}}$, and suspend execution.
        \item Send $\langle \mathsf{TimeStart}, \mathsf{sid},\mathsf{pid}, \mathsf{phase}_p, \delta \rangle$ to $\mathcal{G}_{\mathsf{TIME}}$.
        \item Upon receiving $\langle \mathsf{TimeOver}, \mathsf{sid},\mathsf{pid}, \mathsf{phase}_p, \sigma \rangle$ from $\mathcal{G}_{\mathsf{TIME}}$, resume execution.
        \item Send $\langle \mathsf{Verify}, \mathsf{sid},\mathsf{pid}, m( \mathsf{PRECOMMIT}, h_p, \mathsf{round}_p, \mathsf{id}(v) ), \mathsf{sig} ,\mathsf{pk} \rangle$ to $\mathcal{F}_{\mathsf{SIG}}$ to verify the message signature and obtain $\mathsf{result}_{sig}$.
        \item Send $\langle \mathsf{Verify}, \mathsf{sid},\mathsf{pid}, \mathsf{Validator}(h_p, \mathsf{round}_p), v \rangle$ to $\mathcal{F}_{\mathsf{AUTH}}$ to verify the Proposer's identity and obtain $\mathsf{result}_{auth}$.
        \item If $\mathsf{result}_{sig} =0\lor \mathsf{result}_{auth} = 0$:
        \begin{itemize}
            \item Remove $\mathsf{Validator}(h_p, \mathsf{round}_p)$ from Validator set.
        \end{itemize}
        \item If $\mathsf{valid}(v) \land (\mathsf{count}_{\mathsf{precommit}} > 2f+1)$ and no $\langle \mathsf{TimeOver}, \mathsf{sid},\mathsf{pid}, \mathsf{phase}_p , \delta\rangle$ has been received from $\mathcal{G}_{\mathsf{TIME}}$:
        \begin{itemize}
            \item Set $\mathsf{decision}_p[h_p] = v$.
            \item Send $\langle \mathsf{RoundOK}\rangle$ to $\mathcal{F}_{\mathsf{SYNC}}$.
            \item Send $\langle \mathsf{RequestRound}\rangle$ to $\mathcal{F}_{\mathsf{SYNC}}$, receive its response $d_i$:
            \begin{itemize}
                \item If $d_i=0$, Update $\mathsf{phase}_p \gets \mathsf{propose}$
				\item Reset $\mathsf{count}_{\mathsf{phase}_p},$ and $\mathsf{lockedRound}_p, \mathsf{lockedValue}_p, \\ \mathsf{validRound}_p, \mathsf{validValue}_p$.
				\item Then send $\langle \mathsf{NEWHEIGHT}, h_p, \mathsf{round}_p, \ast \rangle$ to $\mathcal{S}$.
        		\item Otherwise re-execute this step.
            \end{itemize}
        \end{itemize}
        \item Otherwise:
        \begin{itemize}
            \item Update $\mathsf{phase}_p \gets \mathsf{propose}$ and $\mathsf{round}_p \gets \mathsf{round}_p + 1$.
            \item Send $\langle \mathsf{NEWROUND}, h_p, \mathsf{round}_p, \ast \rangle$ to $\mathcal{S}$.
        \end{itemize}
    \end{itemize}
\end{itemize}
\textbf{Round Advance} \\
Upon receiving message $\langle\ast, h_p, \mathsf{round}, \ast, \ast\rangle$:
\begin{itemize}
    \item Set $\mathsf{count}_{\mathsf{nextround}} \gets \mathsf{count}_{\mathsf{nextround}} + 1$.
    \item If $(\mathsf{count}_{\mathsf{nextround}} > f+1) \land \mathsf{round} > \mathsf{round}_p$:
    \begin{itemize}
        \item Send $\langle \mathsf{NEWROUND}, h_p, \mathsf{round}, \ast \rangle$ to $\mathcal{S}$.
    \end{itemize}
\end{itemize}
\end{myframebox}
\vspace{1ex}
\captionof{figure}{The Functionality $\mathcal{F}_{\mathsf{Tendermint}}$}
\label{fig:func-tendermint}
\end{center}

\clearpage
\newpage

\section{UC-Based Protocol Construction}\label{sec:protocol}

\subsection{Formal Description of the Real-World Protocol}

\subsubsection{Modular Subroutines of Protocol $\pi$}

To simplify the formal representation and facilitate modular modeling, we extract and encapsulate multiple subroutines (or subprotocols) from the Tendermint protocol. Each of these subroutines provides operations that primarily involve the generation and forwarding of messages across different protocol phases, as well as the management of local consensus state. We use the notation $\mathcal{\rho}$ to represent a subroutine, so now $\mathcal{\pi}$ is composed of a main protocol and a set of such subprotocols. This abstraction fosters the development of modular composition and composable security analysis within the UC framework, with a particular emphasis on termination under network delay attacks. 
\begin{center}
\centering
\begin{myframebox}[title=Proposal Subroutine  $\mathcal{\rho}_{\scriptscriptstyle \mathsf{PROPOSAL}}$]
Initialization: $Proposal \gets \perp$, $Round \gets 0$.
\begin{itemize}
    \item Upon receiving a $\langle \mathsf{startProposal} \rangle$ message:
        \begin{itemize}
            \item Elect proposer via round-robin: $Proposer \in H$ where $H \subseteq V$ (honest validators).
            \item Initialize voting power: $votingPower_i \gets stake_i$, $\forall i \in \{1,\dots,N\}$.
            \item Update voting power:
                \begin{itemize}
                    \item Non-selected: $votingPower_i \gets votingPower_i + stake_i$.
                    \item Selected: $votingPower_i \gets votingPower_i - \sum_{j \neq i} stake_j$.
                \end{itemize}
            \item Advance round: $Round \gets Round + 1$.
        \end{itemize}
    \item Upon receiving a $\langle \mathsf{timeout}, T \rangle$ message from the adversary $\mathcal{A}$:
        \begin{itemize}
            \item If $T$ valid: $Round \gets Round + 1$ and elect new proposer.
        \end{itemize}
\end{itemize}
\end{myframebox}
\captionof{figure}{Proposal Subroutine  $\mathcal{\rho}_{\scriptscriptstyle \mathsf{PROPOSAL}}$}
\label{fig:func-proposal}
\end{center}

The proposal subroutine $\mathcal{\rho}_{\scriptscriptstyle \mathsf{PROPOSAL}}$, which governs proposer election and proposal dissemination (see Fig.~\ref{fig:func-proposal}), is initialized by setting the proposal to $\perp$ and the round number to 0. Upon receiving a $\mathsf{startProposal}$ message, proposer election occurs via round-robin over honest validators. Each validator’s voting power is initialized to its stake, i.e., for all $i \in \{1, \dots, N\}$, $votingPower_i \gets stake_i$. At the start of each new round, the voting power of the selected proposer decreases by the total stake of the other validators, while the voting power of non-proposers increases by their individual stakes. Upon receiving a $\langle \mathsf{timeout}, T \rangle$ message, if $T$ is valid, the round number is incremented by 1, and a new proposer is elected.

The voting subroutine $\mathcal{\rho}_{\scriptscriptstyle \mathsf{VOTE}}$, which models the \emph{Prevote} and \emph{PreCommit} phases (see Fig.~\ref{fig:func-vote}), initializes a timer (using $\mathcal{G}_{\mathsf{TIME}}$) and defaults to a $\mathsf{nil}$ vote upon timeout. Upon receiving a $\mathsf{Prevote}$ request, it queries the lock status from $\mathcal{\rho}_{\scriptscriptstyle \mathsf{STATE}}$; if the block $B'$ is locked, it broadcasts $\langle v_i, \mathsf{prevote}, \mathsf{Vote}(B') \rangle$. If no block is locked, it broadcasts $\langle v_i, \mathsf{prevote}, \mathsf{Vote}(B) \rangle$ for the current proposal, or $\langle v_i, \mathsf{prevote}, \mathsf{Vote}(\mathsf{nil}) \rangle$ if unavailable. In the \emph{PreCommit} phase, if more than $2f+1$ $\mathsf{prevote}$ messages for block $B$ are received, it broadcasts $\langle v_i, \mathsf{precommit}, \mathsf{Vote}(B) \rangle$, unlocks any  previous lock by sending $\langle v_i, \mathsf{unlock}, B' \rangle$ to $\mathcal{\rho}_{\scriptscriptstyle \mathsf{STATE}}$, and locks block $B$ by sending $\langle v_i, \mathsf{lock}, B \rangle$ to $\mathcal{\rho}_{\scriptscriptstyle \mathsf{STATE}}$. If \text{PREVOTE}(\text{nil}) exceeds the threshold, it broadcasts $\langle v_i, \mathsf{precommit}, \mathsf{Vote}(\mathsf{nil}) \rangle$ and releases all locks by sending $\langle v_i, \mathsf{unlock}, \mathsf{ALL} \rangle$ to $\mathcal{\rho}_{\scriptscriptstyle \mathsf{STATE}}$; otherwise, no lock operation is performed.

\begin{center}
\centering
\begin{myframebox}[title=Vote Subroutine  $\mathcal{\rho}_{\scriptscriptstyle \mathsf{VOTE}}$]
Initialization: Send $ \langle \mathsf{TimeStart}, \delta \rangle $ to $ \mathcal{G}_{\mathsf{TIME}} $. Upon any $ \langle \mathsf{TimeOver} \rangle $ from $ \mathcal{G}_{\mathsf{TIME}} $, vote nil block immediately.
\begin{itemize}
    \item Upon receiving a $ \langle \mathsf{Prevote}, \mathsf{Proposal} \rangle $ message from validator $ v_i \in V $:
        \begin{itemize}
            \item \textit{With Proposal}: 
                \begin{itemize}
                    \item Query $ \mathcal{\rho}_{\scriptscriptstyle \mathsf{STATE}} $ for PoLC.
                    \item If locked to previous $ B' $: broadcast $ \langle v_i, \mathsf{prevote}, \mathsf{Vote}(B') \rangle $.
                    \item Else: broadcast $ \langle v_i, \mathsf{prevote}, \mathsf{Vote}(B) \rangle $.
                \end{itemize}
            \item \textit{No Proposal}: Broadcast $ \langle v_i, \mathsf{prevote}, \mathsf{Vote}(\mathsf{nil}) \rangle $.
        \end{itemize}
    \item Upon receiving a $ \langle \mathsf{PreCommit}, \mathsf{Proposal} \rangle $ message from validator $ v_i \in V $:
        \begin{itemize}
            \item \textit{If $ \geq 2f+1 $ prevotes for $ B $}: 
                \begin{itemize}
                    \item Broadcast $ \langle v_i, \mathsf{precommit}, \mathsf{Vote}(B) \rangle $.
                    \item Send $ \langle v_i, \mathsf{unlock}, B' \rangle $ and $ \langle v_i, \mathsf{lock}, B \rangle $ to $ \mathcal{\rho}_{\scriptscriptstyle \mathsf{STATE}} $.
                \end{itemize}
            \item \textit{If $ \geq 2f+1 $ nil prevotes}: 
                \begin{itemize}
                    \item Broadcast $ \langle v_i, \mathsf{precommit}, \mathsf{Vote}(\mathsf{nil}) \rangle $.
                    \item Send $ \langle v_i, \mathsf{unlock}, \mathsf{ALL} \rangle $ to $ \mathcal{\rho}_{\scriptscriptstyle \mathsf{STATE}} $.
                \end{itemize}
            \item \textit{Otherwise}: no lock operation.
        \end{itemize}
\end{itemize}
\end{myframebox}
\captionof{figure}{Vote Subroutine  $\mathcal{\rho}_{\scriptscriptstyle \mathsf{VOTE}}$}
\label{fig:func-vote}
\end{center}

The commit subroutine $\mathcal{\rho}_{\scriptscriptstyle \mathsf{COMMIT}}$ (see Fig.~\ref{fig:func-commit}) , governs the block finalization process. Each validator maintains a local indicator $c_i$ to record its commitment status. When a $\mathsf{Commit}$ message is received, the subroutine first checks whether at least $2f+1$ $\mathsf{precommit}$ votes support the proposed block $B$. If this threshold is satisfied, the validator broadcasts $\langle v_i, \mathsf{commit}, \mathsf{Vote}(B) \rangle$ and begins collecting $\mathsf{commit}$ votes. Once $2f+1$ $\mathsf{commit}$ votes are gathered, it sets $c_i := 1$, signals successful finalization to the validator via $\mathsf{allowCommit}$, and instructs $\rho_{\scriptscriptstyle \mathsf{STATE}}$ to advance the blockchain height through a $\mathsf{newHeight}$ message. If the threshold is not met, the subroutine rejects the commitment and triggers $\mathsf{newRound}$ to continue the protocol. Additionally, upon receiving a $\mathsf{request_status}$ query, it returns both the current commitment status set $C$ and the finalization result of block $B$. This mechanism ensures that finalization only occurs under sufficient consensus, while providing transparency of status to all validators.

\begin{center}
\centering
\begin{myframebox}[title=Commit Subroutine  $\mathcal{\rho}_{\scriptscriptstyle \mathsf{COMMIT}}$]
Initialization: For each $v_i \in V$, initialize $c_i \gets 0$ (commit status indicator). Send $ \langle \mathsf{TimeStart}, \delta \rangle $ to $ \mathcal{G}_{\mathsf{TIME}} $. Upon any $ \langle \mathsf{TimeOver} \rangle $ from $ \mathcal{G}_{\mathsf{TIME}} $: send $ \langle \mathsf{newRound} \rangle $ to $ \rho_{\scriptscriptstyle \mathsf{STATE}} $.
\begin{itemize}
    \item Upon receiving a $ \langle \mathsf{Commit}, \mathsf{Proposal} \rangle $ message from validator $v_i \in V$:
        \begin{itemize}
            \item If $ \geq 2f+1 $ precommit votes received:
                \begin{itemize}
                    \item Broadcast $ \langle v_i, \mathsf{commit}, \mathsf{Vote}(B) \rangle $.
                    \item Collect network commit votes.
                    \item If $ \geq 2f+1 $ commit votes collected: set $c_i \gets 1$, and send $ \langle \mathsf{allowCommit}, \mathsf{Proposal} \rangle $ to $v_i$ and $ \langle \mathsf{newHeight} \rangle $ to $ \mathcal{\rho}_{\scriptscriptstyle \mathsf{STATE}} $.
                    \item Else: send $ \langle \mathsf{rejectCommit}, \mathsf{Proposal} \rangle $ to $v_i$ and $ \langle \mathsf{newRound} \rangle $ to $ \mathcal{\rho}_{\scriptscriptstyle \mathsf{STATE}} $.
                \end{itemize}
            \item Else: send $ \langle \mathsf{newRound} \rangle $ to $ \mathcal{\rho}_{\scriptscriptstyle \mathsf{STATE}} $.
        \end{itemize}
        
    \item Upon receiving a $ \langle \mathsf{request\_status} \rangle $ message from any party $v_k$:
        \begin{itemize}
            \item Return status set $C$ and finalization status of $B$.
        \end{itemize}
\end{itemize}
\end{myframebox}
\captionof{figure}{Commit Subroutine  $\mathcal{\rho}_{\scriptscriptstyle \mathsf{COMMIT}}$}
\label{fig:func-commit}
\end{center}

\begin{center}
\centering
\begin{myframebox}[title=State Subroutine $\mathcal{\rho}_{\scriptscriptstyle \mathsf{STATE}}$]
Initialization: $Height \gets 0$, $Round \gets 0$, $PoLC \gets \perp$.
\begin{itemize}
    \item Upon receiving a $ \langle \mathsf{newHeight} \rangle $ message from any validator $v_i \in V$:
        \begin{itemize}
            \item $Height \gets Height + 1$, $Round \gets 0$.
        \end{itemize}
    \item Upon receiving a $ \langle \mathsf{newRound} \rangle $ message from any validator $v_i \in V$:
        \begin{itemize}
            \item $Round \gets Round + 1$.
        \end{itemize}
    \item Upon receiving a $ \langle \mathsf{getProposal}, \mathsf{sid}, \mathsf{phase}_p, \ast \rangle $ message from the proposer:
        \begin{itemize}
            \item Retrieve proposals from the configuration file, and send $\langle \mathsf{proposalRec}, \mathsf{sid}, \mathsf{phase}_p, \mathsf{Proposals} \rangle$  to Proposer to the caller.
        \end{itemize}
    \item Upon receiving a $ \langle \mathsf{updateProposal}, \mathsf{sid}, \mathsf{phase}_p, \newline \mathsf{Proposals} \rangle $ message from the proposer:
        \begin{itemize}
            \item Update proposals in the configuration file.
        \end{itemize}
    \item Upon receiving a $ \langle v_i, \mathsf{lock}, B \rangle $ message from validator $v_i$:
        \begin{itemize}
            \item Add $v_i$ to $\langle \mathsf{Height}, \mathsf{Round}, B \rangle$ validator set in $PoLC$.
        \end{itemize}
    \item Upon receiving a $ \langle v_i, \mathsf{unlock}, B \rangle $ message from validator $v_i$:
        \begin{itemize}
            \item Remove $v_i$ from $\langle \mathsf{Height}, \mathsf{Round}, B \rangle$ validator set in $PoLC$.
        \end{itemize}
    \item Upon receiving a $ \langle v_i, \mathsf{unlock}, \mathsf{ALL} \rangle $ message from validator $v_i$:
        \begin{itemize}
            \item Reset $PoLC \gets \perp$.
        \end{itemize}
    \item Upon receiving a $ \langle v_i, \mathsf{queryState} \rangle $ message from validator $v_i$:
        \begin{itemize}
            \item Return current $PoLC$.
        \end{itemize}
\end{itemize}
\end{myframebox}
\captionof{figure}{State Subroutine $\mathcal{\rho}_{\scriptscriptstyle \mathsf{STATE}}$}
\label{fig:func-state}
\end{center}

The state management subroutine $\mathcal{\rho}_{\scriptscriptstyle \mathsf{STATE}}$ (see Fig.~\ref{fig:func-state}) maintains $\mathsf{Height}$, $\mathsf{Round}$, and a PoLC structure that logs, for each $\langle \mathsf{Height}, \mathsf{Round}, B \rangle$, any $\mathsf{PREVOTE}$ sets with size above two-thirds of the number of nodes, so as to track $\mathsf{lockedValue}$ / $\mathsf{lockedRound}$ for current locks and $\mathsf{validValue} / \mathsf{validRound}$ for the latest majority-supported proposal. Upon receiving a $\mathsf{newHeight}$ message from any validator $v_i \in V$, $\mathcal{\rho}_{\scriptscriptstyle \mathsf{STATE}}$ increments the height and resets the round to $0$. Upon receiving a $\mathsf{newRound}$ message, it increments the round. It also responds to proposer queries by retrieving proposals from the configuration file and returning them to the caller. When receiving a $\mathsf{updateProposal}$ message from the proposer, it updates the proposals in the configuration file. Upon receiving a $\mathsf{lock}$ message from a validator $v_i$, it adds $v_i$ to the $\langle \mathsf{Height}, \mathsf{Round}, B \rangle$ validator set in the PoLC structure. Similarly, receiving a $\mathsf{unlock}$ message removes $v_i$ from the set. If the $\mathsf{unlock}$ message indicates $\mathsf{ALL}$, it resets the PoLC structure. Additionally, it enables validators to query the current state of PoLC, thereby informing their voting decisions.


\newlength{\colone}\setlength{\colone}{0.5cm}
\newlength{\coltwo}\setlength{\coltwo}{4.5cm}
\newlength{\colthree}\setlength{\colthree}{5.5cm}
\newlength{\colfour}\setlength{\colfour}{0.5cm}
\newlength{\colspanwidth}
\setlength{\colspanwidth}{\dimexpr\colthree+\colfour\relax}

\begin{center}
\begin{myframebox}[title=The Protocol $\mathcal{\pi}_{\mathsf{Tendermint}}$]
\small 
\renewcommand{\arraystretch}{0.8} 
\setlength{\parskip}{-1pt} 
\tablefirsthead{%
  \hline
  \makebox[\colone][c]{$\mathcal{Z}$} &
  \makebox[\coltwo][c]{Proposer} &
  \makebox[\colthree][c]{Validator} &
  \makebox[\colfour][c]{$\mathcal{A}$} \\
  \hline\\[0.1pt]
}
\tablehead{\hline
  \makebox[\colone][c]{$\mathcal{Z}$} &
  \makebox[\coltwo][c]{Proposer} &
  \makebox[\colthree][c]{Validator} &
  \makebox[\colfour][c]{$\mathcal{A}$} \\
  \hline
}
\tabletail{\hline \multicolumn{4}{r}{Continued on next page} \\}
\tablelasttail{\hline}
\begin{supertabular}{%
    p{\colone}%
    p{\coltwo}%
    p{\colthree}%
    p{\colfour}%
}
\begin{tikzpicture}[baseline]
      \draw[->, thick] (0,0) -- (\linewidth,0);
\end{tikzpicture} & 1: Send $\langle \mathsf{TimeStart}$, $ \mathsf{sid}$, $\mathsf{phase}_p$, $ \delta \rangle$ to $\mathcal{G}_{\mathsf{TIME}}$ &  &  \\\\
& 2: Send \newline $\langle\mathsf{getProposal},  \mathsf{sid},  \mathsf{phase}_p, \ast \rangle$ to $\mathcal{\rho}_{\scriptscriptstyle \mathsf{STATE}}$  &  &  \\\\
& 3: Get $\langle \mathsf{proposalRec}, \mathsf{sid},  \mathsf{phase}_p, \newline  \mathsf{Proposals} \rangle$ from $\mathcal{\rho}_{\scriptscriptstyle \mathsf{STATE}}$ &  &  \\\\
& 4: Select a Proposal value $v$ from the Proposals. &  \\\\
& 5: Send \newline $\langle \mathsf{Verify}, \mathsf{sid},  m(\mathsf{PROPOSAL}, h_p, \newline \mathsf{round}_p, v), \mathsf{sig},\mathsf{pk} \rangle$ to $\mathcal{F}_{\mathsf{SIG}}$ &  &  \\\\
& 6: Send \newline $\langle \mathsf{Verify}, \mathsf{sid}, \mathsf{Proposer}(h_p, \mathsf{round}_p),   v \rangle$ to $\mathcal{F}_{\mathsf{AUTH}}$ &  &  \\\\
& 7: If $\mathsf{result}_{sig} =0 \lor \mathsf{result}_{auth} = 0$, remove $\mathsf{Proposer}(h_p, \mathsf{round}_p)$  \newline  and call $\mathcal{\rho}_{\scriptscriptstyle \mathsf{PROPOSAL}}$ &  &  \\\\
& 8: Send $\langle \mathsf{Sleep}, \mathsf{sid},  h_p, \mathsf{round}_p, v \rangle$ to $\mathcal{A}$
    & \multicolumn{2}{>{\centering\arraybackslash}p{\colspanwidth}}{%
        \begin{tikzpicture}[baseline]
          \draw[->, thick] (0,0) -- (\linewidth,0);
        \end{tikzpicture}%
      } \\\\
& 9: Get $\langle \mathsf{Wake}, \mathsf{sid}, h_p, \mathsf{round}_p, v \rangle$ from $\mathcal{A}$ 
    & \multicolumn{2}{>{\centering\arraybackslash}p{\colspanwidth}}{%
        \begin{tikzpicture}[baseline]
          \draw[<-, thick] (0,0) -- (\linewidth,0);
        \end{tikzpicture}%
      } \\\\
& 10: If $\mathsf{valid}(v)$ and get $\langle \mathsf{TimeOver},  \mathsf{sid}, \mathsf{phase}_p, \newline  \delta \rangle$, call $\mathcal{\rho}_{\scriptscriptstyle \mathsf{PROPOSAL}}$ &  &  \\\\
& 11: Otherwise, broadcast $\langle \mathsf{PROPOSAL}, h_p,   \mathsf{round}_p, \newline  v \rangle$ to Validator &  &  \\\\
&  & 12: Send $\langle \mathsf{TimeStart}, \mathsf{sid},  \mathsf{phase}_p,  \delta \rangle$ to $\mathcal{G}_{\mathsf{TIME}}$ &  \\\\
&  & 13: If $\mathsf{valid}(v)$, call $\mathcal{\rho}_{\scriptscriptstyle \mathsf{VOTE}}\langle \mathsf{Prevote}, \newline  \mathsf{PROPOSAL} \rangle$ &
\begin{tikzpicture}[baseline]
  \def\yup{4pt}
  \def\ydown{-4pt}
  \draw[->, thick] (0,\yup) -- (\linewidth,\yup);
  \draw[<-, thick] (0,\ydown) -- (\linewidth,\ydown);
\end{tikzpicture} \\\\
&  & 14: If $\mathsf{valid}(v)$, call $\mathcal{\rho}_{\scriptscriptstyle \mathsf{VOTE}}\langle \mathsf{PreCommit}, \newline  \mathsf{PROPOSAL} \rangle$ &
\begin{tikzpicture}[baseline]
  \def\yup{4pt}
  \def\ydown{-4pt}
  \draw[->, thick] (0,\yup) -- (\linewidth,\yup);
  \draw[<-, thick] (0,\ydown) -- (\linewidth,\ydown);
\end{tikzpicture} \\\\
&  & 15: If $\mathsf{valid}(v)$, call $\mathcal{\rho}_{\scriptscriptstyle \mathsf{COMMIT}}\langle \mathsf{Commit}, \newline \mathsf{PROPOSAL} \rangle$ &
\begin{tikzpicture}[baseline]
      \draw[->, thick] (0,0) -- (\linewidth,0);
\end{tikzpicture} \\\\
\begin{tikzpicture}[baseline]
      \draw[<-, thick] (0,0) -- (\linewidth,0);
\end{tikzpicture} & Output $\langle \mathsf{Success}, \mathsf{sid}, \newline id(v) \rangle$  to $\mathcal{Z}$ &  16: If a message $\langle \mathsf{allowCommit}, \newline \mathsf{PROPOSAL} \rangle$ is received from $\mathcal{\rho}_{\scriptscriptstyle \mathsf{COMMIT}}$ and no $\langle \mathsf{TimeOver} \rangle$, call $\mathcal{\rho}_{\scriptscriptstyle \mathsf{PROPOSAL}}$ &
\begin{tikzpicture}[baseline]
      \draw[<-, thick] (0,0) -- (\linewidth,0);
\end{tikzpicture} \\\\
\begin{tikzpicture}[baseline]
      \draw[<-, thick] (0,0) -- (\linewidth,0);
\end{tikzpicture} & Output $\langle \mathsf{Failure}, \mathsf{sid}, \newline  \perp \rangle$ to $\mathcal{Z}$ &  17: Otherwise, call  $\mathcal{\rho}_{\scriptscriptstyle \mathsf{PROPOSAL}} \newline \langle \mathsf{NewRound}, h_p, \mathsf{round}_p+1 \rangle$ &  \\\\
\end{supertabular}
\end{myframebox}
\captionof{figure}{The Protocol $\mathcal{\pi}_{\mathsf{Tendermint}}$}
\label{fig:protocol-tendermint}
\end{center}

\subsubsection{The Protocol $\pi_{\mathsf{Tendermint}}$}

In this work, we present a formal model of the real-world Tendermint protocol, denoted as $\pi_{\mathsf{Tendermint}}$, with its execution flow illustrated in Fig.~\ref{fig:protocol-tendermint}. We represent the protocol as a multiparty interactive system involving proposers, validators, the environment $\mathcal{Z}$, and the adversary $\mathcal{A}$, enabling rigorous analysis of its resilience under adversarial network delays through interactions with ideal functionalities.

For the proposer, the model captures initiation of phase-specific timers via $\mathcal{G}_{\mathsf{TIME}}$, acquisition of proposal values through $\mathcal{\rho}_{\scriptscriptstyle \mathsf{STATE}}$, and validation using $\mathcal{F}_{\mathsf{SIG}}$ and $\mathcal{F}_{\mathsf{AUTH}}$. Invalid proposals trigger $\mathcal{\rho}_{\scriptscriptstyle \mathsf{PROPOSAL}}$ to advance rounds. Network latency is formalized by adversarial $ \mathsf{Sleep}/\mathsf{Wake} $ controls, which may conflict with timeouts from $\mathcal{G}_{\mathsf{TIME}}$. Prolonged delays force round failure, while timely wakeups allow broadcasting $ \mathsf{PROPOSAL}$.

For validators, upon receiving a proposal, timers are set via $\mathcal{G}_{\mathsf{TIME}}$, and voting proceeds through $\mathcal{\rho}_{\scriptscriptstyle \mathsf{VOTE}}$ to issue \textsf{Prevote} and \textsf{PreCommit} messages, culminating in $\mathcal{\rho}_{\scriptscriptstyle \mathsf{COMMIT}}$ for finalization. The adversary may observe and delay these interactions. Consensus is successful if a validator receives an $ \mathsf{allowCommit} $ signal before timeout, producing a $ \mathsf{Success} $ message for $\mathcal{Z}$ and advancing to the next height. Otherwise, expiration or invalidity triggers $\mathcal{\rho}_{\scriptscriptstyle \mathsf{PROPOSAL}}$ to restart the round and reports $ \mathsf{Failure} $, thereby ensuring eventual termination.

\subsection{Mapping from the Protocol to the Ideal Functionality}

\begin{theorem}[$\pi_{\mathsf{Tendermint}}$ GUC-Realizes $\mathcal{F}_{\mathsf{Tendermint}}^{V,\Delta,\delta,\tau}$]
The real-world protocol $\pi_{\mathsf{Tendermint}}$ GUC-realizes the ideal functionality $\mathcal{F}_{\mathsf{Tendermint}}^{V,\Delta,\delta,\tau}$ under adversarial network-delay attacks.
\label{thm:tendermint-mapping}
\end{theorem}

\begin{proof}
To prove that $\pi_{\mathsf{Tendermint}}$ GUC-realizes $\mathcal{F}_{\mathsf{Tendermint}}$, we construct a simulator $\mathcal{S}$ that, by interacting with the ideal functionality $\mathcal{F}_{\mathsf{Tendermint}}$ and the global timer functionality $\mathcal{G}_{\mathsf{TIME}}$, simulates the interaction between the real-world protocol $\pi_{\mathsf{Tendermint}}$ and the adversary $\mathcal{A}$. The construction must ensure that no environment $\mathcal{Z}$ can distinguish between the interaction in the ideal world and the execution of the protocol in the real world. The simulator's main task is to translate the adversary's actions in the real world, particularly its control over network delays, into valid instructions for the ideal functionality, thereby creating a view for the environment that is consistent with the protocol’s rules.

The simulator $\mathcal{S}$ runs dummy parties of $\pi_{\mathsf{Tendermint}}$ for each honest party, intercepting adversarial messages and forwarding them to the corresponding instances. Time is managed via the ideal timer functionality $\mathcal{G}_{\mathsf{TIME}}$, rather than local clocks. At the start of each phase in round $r$, $\mathcal{S}$ computes the timeout $\delta_r = \tau^{\mathsf{init}}_{\mathsf{phase}} + r \cdot \tau^{\mathsf{step}}_{\mathsf{phase}}$ and issues a $\langle \mathsf{TimeStart}, \mathsf{sid}, \mathsf{phase}_p, \delta_r \rangle$ command. If $\mathcal{A}$ supplies quorum messages (e.g., $2f+1$ PreVotes) before timeout, $\mathcal{S}$ advances the simulation, resets the timer, and updates $\mathcal{F}_{\mathsf{Tendermint}}$. Otherwise, upon $\langle \mathsf{TimeOver}, \mathsf{sid}, \mathsf{phase}_p, \delta \rangle$ from $\mathcal{G}_{\mathsf{TIME}}$, $\mathcal{S}$ simulates honest timeout behavior (nil votes, round change) and informs $\mathcal{F}_{\mathsf{Tendermint}}$.

Indistinguishability follows since $\mathcal{S}$ translates adversarial delays into valid protocol-driven timeouts, leaving $\mathcal{Z}$ unable to discern whether a round change arises from adversarial actions or network latency. Security is ensured by locking, threshold voting, escalating timeouts, and BFT limits. Termination is guaranteed since $\mathcal{G}_{\mathsf{TIME}}$ enforces increasing timeouts, eventually enabling quorum despite bounded delays. Thus, $\pi_{\mathsf{Tendermint}}$ GUC-realizes the ideal functionality.
\end{proof}

\noindent Therefore, it is established that $\pi_{\mathsf{Tendermint}}$ correctly implements the ideal functionality $\mathcal{F}_{\mathsf{Tendermint}}$ under network delay attacks, proving Theorem \ref{thm:tendermint-mapping}.

\section{UC Termination Proof for Tendermint}
\label{sec:securityProof}

\begin{theorem}[Tendermint Protocol UC Termination]
Let the total number of nodes be $n$ with at most $f$ Byzantine nodes ($n \geq 3f+1$). After Global Stabilization Time (GST), the system operates in a partially synchronous model with maximum network delay $\Delta$. Define:

\begin{equation}
T^\ast = \mathsf{GST} + O(f^2\Delta)
\end{equation}

Then for any block height $h$ and any honest process $p \in \mathcal{H}$, when system time $t \geq T^\ast$, protocol $\pi_{\mathsf{Tendermint}}$ terminates at height $h$ and outputs a unique decision value $v_h$: $\mathsf{decide}_h(p) = v_h $ and $ v_h $ is consistent across all $ p \in \mathcal{H}.$
\label{thm:tendermint-termination}
\end{theorem}

\subsection{Tendermint Protocol Features and Adversary Model}

The core of Tendermint is to achieve Byzantine fault tolerance through a five-phase consensus process: $\emph{NewRound}$, $\emph{Propose}$, $\emph{Prevote}$, $\emph{PreCommit}$, and $\emph{Commit}$. Its defense against delay attacks relies on a set of timeouts that grow linearly:

\begin{equation}
\tau^{(r)}_{\text{phase}} = \tau^{\mathsf{init}}_{\text{phase}} + r \cdot \tau^{\mathsf{step}}_{\text{phase}}
\end{equation}

\noindent where $r$ is the current round and $\tau^{\mathsf{init}}, \tau^{\mathsf{step}}$ are configurable thresholds. To preserve consistency, Tendermint employs state-locking via \text{lockedRound} and \text{validRound}, ensuring that honest validators reject conflicting proposals. Proposal rights are rotated through a weighted round-robin mechanism, where each validator’s voting power is dynamically updated as $\mathsf{votingPower}_i = \mathsf{stake}_i - \sum_{j \neq i} \mathsf{stake}_j$, thereby preventing centralized control and promoting fairness.

The adversary $\mathcal{A}$ may corrupt up to $f$ validators and one proposer, enabling it to inject arbitrary proposals $\langle \mathsf{PROPOSAL}, h_p, \mathsf{round}_p, v \rangle$, forge votes $\langle \mathsf{Vote.Phase}, h_p, \mathsf{round}_p, \mathsf{Vote.Result} \rangle$, and introduce targeted message delays through $\langle \mathsf{Sleep/Wake}, \mathsf{sid}, \sigma_{\mathsf{adv}} \rangle$ commands. The model enforces three constraints: (1) adversaries cannot forge signatures from honest parties; (2) consensus halts if active validators drop below $2f+1$; and (3) Tendermint guarantees eventual termination through timeouts and nil-vote escalation under delay attacks.

\subsection{Termination Proof}

In this section, we show that, under the UC framework, real and ideal executions remain indistinguishable, thereby ensuring Tendermint’s termination guarantee. The complete termination proof of the Tendermint consensus protocol under the UC framework is as follows. To structure the proof, we analyze the protocol's behavior under three mutually exclusive scenarios based on adversarial actions and the protocol's timeout mechanism. We will demonstrate that in each case, the simulator can perfectly mimic the real-world outcomes, thus upholding the indistinguishability claim and confirming Tendermint's termination property. 

\subsubsection{Case 1: Non-Timeout Scenario ($\delta+\sigma_{\mathsf{adv}} \leq \tau$)}

\paragraph{Corrupted Proposer Only}

When the adversary corrupts only the proposer, it delays the proposal procedure utilizing invoking \(\mathcal{F}_{\mathsf{BC}}\) to send a message \(\langle \mathsf{PROPOSAL}, \\ h_p, \mathsf{round}_p, v \rangle\) at time \(t_0 + \sigma_{\mathsf{adv}}\), where \(\sigma_{\mathsf{adv}} \leq \tau\). Honest parties receive the proposal within \([t_0, t_0 + \Delta]\), proceed to validate and broadcast $\mathsf{PREVOTE}$ messages by \(t_0 + 2\Delta\), collect \(\geq 2f+1\) votes and issue $\mathsf{PRECOMMIT}s$ by \(t_0 + 3\Delta\), and finalize via $\mathsf{COMMIT}$ by \(t_0 + 4\Delta\). In the ideal world, the simulator \(\mathcal{S}\) mimics this delayed broadcast and reproduces all protocol phases with matching timing. Since all proposal, voting, and commit operations occur within the same time windows and follow identical rules, the environment \(\mathcal{Z}\) observes indistinguishable executions across real and ideal worlds. Thus, termination is guaranteed when only the proposer is corrupted and \(\sigma_{\mathsf{adv}} \leq \tau\).

\paragraph{Corrupted Proposer and Validators}

When the adversary corrupts the proposer and up to \(k \leq f-1\) validators, it broadcasts a delayed proposal via \(\mathcal{F}_{\mathsf{BC}}\) at time \(t_0 + \sigma_{\mathsf{adv}}\) and withholds votes from the corrupted validators. Honest parties receive the proposal within \([t_0, t_0 + \Delta]\), validate it, and enter the Prevote phase. Despite partial vote withholding, honest nodes can still collect \(\geq 2f+1\) votes within the timeout window (as \(\sigma_{\mathsf{adv}} \leq \tau\)), enabling them to complete both Prevote and PreCommit phases on time. All honest nodes finalize the Commit phase by \(t_0 + 4\Delta\), ensuring termination through honest majority participation. In the ideal world, the simulator \(\mathcal{S}\) reproduces these behaviors by synchronizing the delayed proposal broadcast and selectively delaying corrupted validators’ votes, while honest-party actions are faithfully simulated. The resulting protocol execution preserves time and output consistency across both worlds---honest nodes perform identical operations at the same logical timestamps---thus rendering the environment \(\mathcal{Z}\) unable to distinguish between the real and ideal executions. Therefore, Tendermint guarantees termination when the adversary controls the proposer and a minority of validators, provided \(\sigma_{\mathsf{adv}} \leq \tau\).

\paragraph{Corrupted Validators}

When the adversary controls up to \(k \leq f\) validators, termination depends on the honest majority. At time \(t_0\), the honest proposer broadcasts the proposal \(\langle \mathsf{PROPOSAL}, h_p, \mathsf{round}_p, v \rangle\) to all parties, which honest nodes receive within \([t_0, t_0 + \Delta]\) and then proceed to validate and enter the Prevote phase. Corrupted validators delay their votes until \(t_0 + \Delta + \sigma_{\mathsf{adv}}\); however, since \(\sigma_{\mathsf{adv}} \leq \tau\), honest parties can still gather at least \(2f + 1\) votes on time, enabling timely progression through the PreCommit and Commit phases, with final commits completed before \(t_0 + 4\Delta\). The simulator \(\mathcal{S}\) faithfully reproduces adversarial delays on corrupted validators while simulating honest-node compliance, ensuring identical timing and behavior across real and ideal executions. This results in indistinguishability from the environment \(\mathcal{Z}\), which observes matching proposals, vote patterns, and commit outcomes in both worlds. Hence, Tendermint guarantees termination under validator corruption within the adversarial delay bound \(\sigma_{\mathsf{adv}} \leq \tau\).

\subsubsection{Case 2: Timeout-Triggered Scenario ($\delta+\sigma_{\mathsf{adv}} > \tau$)}

\paragraph{Corrupted Proposer Only}  
When the adversary controls the Proposer and delays the proposal broadcast beyond the timeout threshold $\tau$, honest nodes in the real world detect the absence of a timely proposal and promptly advance to the next round, thereby preventing indefinite stalling and ensuring consensus progress. In the ideal world, the simulator $\mathcal{S}$ mimics this behavior by suppressing the adversary-controlled Proposer's broadcast and triggering the round transition synchronously with real-world timing. This synchronization ensures that honest parties' actions and round progressions are indistinguishable between worlds. As a result, the environment $\mathcal{Z}$ observes identical sequences of events, preserving both temporal and output consistency, which guarantees termination and renders the real and ideal executions indistinguishable.

\paragraph{Corrupted Proposer and Validators}  
When the adversary controls both the Proposer and up to $k \leq f-1$ Validators while delaying communication beyond $\tau$, the protocol behavior parallels Case 2(a): honest nodes detect the Proposer's timeout and promptly initiate new rounds. Although corrupted Validators may delay votes, honest parties proceed unaffected due to the enforced timeout mechanisms. The simulator $\mathcal{S}$ reproduces the Proposer timeout and synchronizes adversary-induced vote delays to mirror real-world timing, ensuring honest nodes behave consistently across executions. Consequently, environment $\mathcal{Z}$ observes identical event timing and outputs in both worlds, preserving indistinguishability. Thus, termination is guaranteed, with $\mathcal{S}$ maintaining perfect temporal and output alignment.

\paragraph{Corrupted Validators}  
When the adversary controls up to $k \leq f$ Validators and induces delays exceeding the timeout $\tau$, termination is still ensured by the honest majority. Real-world honest nodes detect these timeouts and advance rounds regardless of adversarial vote delays, relying on the protocol’s timeout mechanisms to maintain progress. Simulator $\mathcal{S}$ faithfully reproduces the adversarial delay patterns while ensuring honest nodes follow the protocol, thereby perfectly mirroring temporal effects and adversarial influence. This alignment prevents environment $\mathcal{Z}$ from distinguishing between real and ideal executions based on timing or outcomes. Consequently, strict temporal and output consistency is preserved: honest nodes perform identical actions and produce matching outputs in both worlds, guaranteeing termination and maintaining indistinguishability despite adversarial Validator delays.

\subsubsection{Case 3: Consensus Completion at Round $r = f+1$}

In the worst-case scenario, Tendermint endures consecutive timeouts for the first $f$ rounds before successfully committing in the first honest round $r = f+1$. Consider a correct process $p$ as the earliest entrant to round $r$, transitioning immediately after the expiration of $\tau^{r-1}_{\mathsf{Precommit}}$. At this point, $p$ has received at least $2f + 1$ PreCommit votes for round $r-1$ by time $t$. By gossip propagation guarantees, all honest processes receive these votes by $t + \Delta$, ensuring their entry into round $r$ by $t + \Delta + \tau^{r-1}_{\mathsf{Precommit}}$. The slowest honest process $q$ enters round $r$ at this time, broadcasting its round-$r$ proposal, which all honest parties receive by $t + 2\Delta + \tau^{r-1}_{\mathsf{Precommit}}$. After validating the proposal $v$, they issue Prevote messages, and upon collecting $2f+1$ Prevotes by $t + 3\Delta + \tau^{r-1}_{\mathsf{Precommit}}$, they trigger PreCommit votes, ensuring convergence on the candidate value.

Subsequently, once sufficient PreCommits are gathered, all processes enter the Commit phase by $t + 4\Delta + \tau^{r-1}_{\mathsf{Precommit}}$, and final commitment occurs upon receiving the required Commit votes by $t + 5\Delta + \tau^{r-1}_{\mathsf{Precommit}}$, terminating the protocol at height $h$ with committed value $v$. In the ideal world, simulator $\mathcal{S}$ replicates this behavior by setting $\tau^{(f+1)}_{\mathsf{phase}}$ to align with the observed round durations, faithfully reproducing timeout transitions and message scheduling. This ensures that the environment $\mathcal{Z}$ cannot distinguish between the ideal and real executions, thereby establishing UC indistinguishability under worst-case adversarial delays.

\subsection{Global Time Analysis}

To thoroughly analyze Tendermint's worst-case global termination time, we consider a scenario where an adversary, controlling $f$ Byzantine replicas, strategically forces the first $f$ rounds ($r=0, 1, \dots, f$) to timeout due to manipulated message delays or faults. Consensus is guaranteed to be reached in round $r=f+1$ because the adversary can no longer prevent quorum formation. The total worst-case time, $T^*$, therefore aggregates the cumulative delays across all four phases (Propose, Prevote, PreCommit, and Commit) of these $f$ failed rounds, as well as the delays incurred during the first successful round ($r=f+1$). Formally, $T^*$ is defined as:

\begin{equation}
\begin{aligned}
T^* = 4\sum_{r=0}^f \tau^{(r)}_{\mathsf{phase}} + 4\tau^{(f+1)}_{\mathsf{phase}}
\end{aligned}
\end{equation}

Here, $\tau^{(r)}_{\mathsf{phase}}$ represents the per-phase timeout in round $r$. Tendermint's dynamic timeout mechanism is designed to ensuretermination, exhibiting linear growth with the round number:

\begin{equation}
\begin{aligned}
\tau_{\mathsf{phase}}^{(r)}
  &= \tau_{\mathsf{phase}}^{\mathsf{init}} + r \cdot \tau_{\mathsf{phase}}^{\mathsf{step}} \\
  &= \Delta + r \cdot \Delta \\
  &= (1 + r) \cdot \Delta
\end{aligned}
\end{equation}
where $\Delta$ is the network one-way delay, accounting for the maximum anticipated network latency for a single message transmission.

Substituting this definition into the $T^*$ formula yields the following closed-form expressions for each component:

\begin{align}
4\sum_{r=0}^f \tau_{\mathsf{phase}}^{(r)}
  &= 2(f+1)(f+2)\,\Delta, \\[6pt]
4\,\tau_{\mathsf{phase}}^{(f+1)}
  &= 4(f+2)\,\Delta.
\end{align}

Summing these components, we derive the simplified worst-case termination time:

\begin{equation}
\begin{aligned}
T^* 
  &= 2(f+1)(f+2)\,\Delta + 4(f+2)\,\Delta \\
  &= 2(f+2)\,\Delta\bigl[(f+1) + 2\bigr] \\
  &= 2(f+2)(f+3)\,\Delta \\
  &= O(f^2\Delta)
\end{aligned}
\end{equation}

The $O(f^2\Delta)$ bound rigorously confirms Tendermint's polynomial termination guarantee even under severe adversarial delay attacks. Crucially, this analysis satisfies indistinguishability requirements in the UC framework, proving Theorem \ref{thm:tendermint-termination}. For any message $m$, the delay discrepancy $\sigma_{\mathsf{delay}}(m)$ between real-world executions and their simulated counterparts is negligible. The simulator $\mathcal{S}$ meticulously reproduces delay attacks by mapping adversarial delays $\sigma_{\mathsf{adv}}$ to timeout-triggered round transitions. This precise mapping prevents the environment $\mathcal{Z}$ from distinguishing between real and ideal executions based on timing or message patterns, thereby establishing UC-security, i.e., termination, for Tendermint under adversarial conditions.

\section{Conclusion}\label{sec:conclusion}

\subsection{Summary}

In this paper, we present the first UC modeling of the Tendermint protocol, and the proof for its termination against network delay attacks in the UC framework. By embedding Tendermint within the UC framework, we distill its core mechanisms, proposal locking, linear timeout progression, and adaptive round advancement, into the ideal functionality $\mathcal{F}_{\mathsf{Tendermint}}^{V,\Delta,\delta,\tau}$. Our simulator accurately captures both honest party delays and adversarial interference, demonstrating that Tendermint guarantees termination with up to $ f < n/3 $ Byzantine faults under $\Delta$-bounded message delays. The analysis further shows that Tendermint achieves a worst-case termination latency of $ O(f^2 \Delta) $ under partial synchrony. This result not only demonstrates the protocol’s theoretical resilience but also bridges the gap between practical performance and formal security guarantees, establishing Tendermint as a robust solution in the context of fault-tolerant consensus.

\subsection{Future Directions}

There are several directions for future work. First, our model assumes a uniform, known delay bound, $\Delta$. Relaxing this assumption to accommodate heterogeneous or mobile-network environments, where delays may vary unpredictably, would enhance the model’s applicability and better reflect real-world conditions. Second, 
it would be interesting to apply the UC modeling methodology of this paper to a broader range of BFT-like protocols. Moreover, integrating Tendermint with higher-level blockchain modules, such as mempools, smart contracts, and cross-chain bridges, would validate its robustness in fully composable, real-world applications, paving the way for modular, scalable and secure blockchain ecosystems.

%
%
%
\bibliographystyle{splncs04}
\bibliography{paper}

\end{document}